\documentclass{article}
\usepackage[utf8]{inputenc}
\usepackage{placeins}
\usepackage{bm}
\usepackage{bbm}
\usepackage{amsmath}
\usepackage{mathtools}
\DeclarePairedDelimiter\floor{\lfloor}{\rfloor}
\usepackage{multirow}
\usepackage{amssymb}
\usepackage{comment}
\usepackage{amsthm}
\usepackage{geometry}
\usepackage{graphicx}
\usepackage{setspace}
\usepackage{hyperref}
\newcommand{\indep}{\rotatebox[origin=c]{90}{$\models$}}

\theoremstyle{definition}
\newtheorem{theorem}{Theorem}[section]

\newtheorem{lemma}{Lemma}[section]
\newtheorem{proposition}{Proposition}[section]
\newtheorem{assumption}{Assumption}

\newtheorem{remark}{Remark}

\title{Sequential change-point detection for compositional time series with exogenous variables}
\author{%
  Yajun Liu \footnote{Amazon.com, Seattle, WA, USA. Email: yajunliustacy@gmail.com}%
  \and Beth Andrews \footnote{Department of Statistics and Data Science, Northwestern University, Evanston, IL, USA. Email: bandrews@northwestern.edu}%
  }
\date{}

\begin{document}

\maketitle

\begin{abstract}
    Sequential change-point detection for time series enables us to sequentially check the hypothesis that the model still holds as more and more data are observed. It is widely used in data monitoring in practice. In this work, we consider sequential change-point detection for compositional time series, time series in which the observations are proportions. For fitting compositional time series, we propose a generalized Beta AR(1) model, which can incorporate exogenous variables upon which the time series observations are dependent.  We show the compositional time series are strictly stationary and geometrically ergodic and consider maximum likelihood estimation for model parameters. We show the partial MLEs are consistent and asymptotically normal and propose a parametric sequential change-point detection method for the compositional time series model. The change-point detection method is illustrated using a time series of Covid-19 positivity rates.
\end{abstract}

\section{Introduction}
\label{sec: introduction}
Change-point detection methods have been well-researched and have drawn researchers' attentions.
In this work, we consider sequential change-point detection for a generalized Beta AR model for describing univariate time series of proportions. In the literature, time series recording proportions of a whole are commonly referred to as compositional time series. In practice, time series recording fluctuating proportions of a whole are commonly seen and attract interest in a lot of fields, like economics, sociology, etc. For example, variations in the proportion of male/female babies born, a univariate compositional time series, during wartime v.s. peacetime, have drawn researchers attention and have been widely investigated, e.g. \cite{chao2019systematic}. In economics, the proportion of stocks with increasing prices, a compositional time series, can reflect the overall economic situation. In this paper, we model univariate compositional time series from the idea that for each time point, the output follows a Beta distribution of which the mean is determined by past information of the series through a generalized linear model. It is named as generalized Beta AR($p$) model. Given the prevalence of compositional time series, it is necessary to develop associated sequential change-point detection methods to detect possible parameter changes in time. 

In Section \ref{sec: Literature Review}, we go through some models for compositional data, especially compositional time series in the literature. Existing change-point detection methods for compositional time series will be discussed in this section. The constraints of existing models are discussed in Section \ref{sec: Literature Review}, which are the motivations for developing a new model. In Section \ref{sec: Model definition}, we first introduce the generalized Beta AR($p$) model, and we prove the strict stationarity and geometric ergodicity of the process when $p=1$. The corresponding parameter estimation will be discussed in Section \ref{sec: Model definition}. In Section \ref{sec: monitoring scheme and null}, we first propose the sequential change-point detection method and the test statistic. Later, asymptotic distribution of the test statistic under the no change point null hypothesis is given and proved in Section \ref{sec: monitoring scheme and null} as well. Power analysis of the test statistic is discussed in Section \ref{sec: power analysis}. In Section \ref{sec: simulation}, we conduct several simulation studies to support the theoretical results derived in the previous sections. We first verify the consistency and asymptotic normality of the estimated parameters. Also, we design two sets of experiment to investigate the behaviors of the test statistic under the null hypothesis and the alternative in Section \ref{sec: simulation}. In Section \ref{sec: real life data}, we discuss two applications of the model and the methodology, respectively. First, we fit generalized Beta AR(1) models to a percentage change of the consumer price index for alcoholic beverages time series and discuss the forecasting ability of the models. Then, we fit a generalized Beta AR(1) model to Covid-19 positivity rates in Arizona and conduct the change-point detection method. Eventually, we successfully detect the change indicating the outbreak of Omicron variant. 

\section{Literature Review}
\label{sec: Literature Review}
Considering the observations in compositional time series are bounded in the unit interval $[0,1]$ (or $(0,1)$), and for multivariate series, the sum-to-one constraint should be satisfied, there are two main directions for modeling compositional time series. 

By applying appropriate transformation schemes, the outputs between 0 and 1 can be mapped to other spaces. Many studies for compositional data are done following this idea. In \cite{atchison1980logistic}, the authors investigate the properties of logistic-normal distributions and use them to analyze compositional data. \cite{aitchison1982statistical} is considered as the pioneering research for compositional data as it summarizes several common transformations for modeling compositional data. In \cite{brunsdon1998time}, the authors apply additive logratio (alr) transformation to build a vector ARMA compositional model; forecasting time series following the model is also developed. Then, based on alr, a state space model for compositional time series is developed in \cite{silva2001modelling}. For multivariate data, relying on alr, a regression model incorporating covariates and VARMA errors is developed to fit the transformed data. Other transformations like Box-Cox transformation, centered logratio (clr) transformation, are also developed. Literature on other transformation schemes are thoroughly introduced in the review paper \cite{larrosa2017compositional}. By mapping the outputs in the unit interval to other spaces, like the real line, we can use existing methods for analyzing time series defined on those spaces to analyze compositional time series. However, without relying on distributions specifically for compositional data, we need to reverse the procedure back to the unit interval for further analysis like forecasting. Meanwhile, correlation structure of compositional data is hard to reveal, as the correlation structure is usually measured based on the transformed data, e.g. \cite{silva2001modelling}.

The other direction is based on distributions naturally defined on $[0,1]$. For univariate series, autoregressive models for which the marginal distributions of the random variables are Beta distributions are built in \cite{mckenzie1985autoregressive}. The marginal distribution of each output follows the same Beta distribution, and the conditional distribution of the current output is a Beta random variable with parameters related to the previous output and an independent Beta innovation directly. The generalized ARMA (GARMA) model is an option to model Beta random variable with time-varying conditional distribution. That is, the time series follows a specific distribution with parameters determined by past information and exogenous variables through a link function. GARMA models of which the conditional distribution of each output based on past information belongs to the same exponential family are proposed in \cite{benjamin2003generalized} (see page 215 (1) and (2)). Following the same idea, in \cite{ferrari2004beta}, the authors build a Beta regression model for compositional time series (see page 801 (4) and (5)). That is, the compositional time series $\{X_t\}$ follows
\begin{equation*}
    f(X_t;\mu_t, \tau) = \frac{1}{B(\tau \mu_t, \tau (1-\mu_t))}X_t^{\tau \mu_t-1}(1-X_t)^{\tau(1-\mu_t)-1}, 0<X_t<1,
\end{equation*}
and the conditional mean $\mu_t$ is determined by the known covariates $\bm{W}_t = (W_{t1}, W_{t2}, ..., W_{tk})$ through a strictly monotonic and twice differentiable link function $g(.)$,
\begin{equation*}
    g(\mu_t) = \sum_{i=1}^k W_{ti}\phi_i := \bm{W}_t^T \bm{\phi}. 
\end{equation*}
The properties of maximum likelihood estimates of $\phi_i$ and $\tau$ and model diagnostics are investigated in \cite{ferrari2004beta}. Based on this work, for compositional time series for which the mean depends on both past observations and exogenous variables, $\beta$ARMA is developed in \cite{rocha2009beta} (page 531). The conditional mean is determined through
\begin{equation*}
    g(\mu_t) = \alpha+\bm{W}_t^T \bm{\phi}+\sum_{i=1}^p \varphi_i\{g(X_{t-i})-\bm{W}_{t-i}^T \bm{\phi}\} + \sum_{j=1}^q \theta_j r_{t-j},
\end{equation*}
where $r_t = g(X_t)-g(\mu_t)$ is the error term defined based on the link function.
Parameter estimation and hypothesis testing for the parameters are also investigated in this work. Goodness-of-fit tests for $\beta$ARMA models are proposed in \cite{scher2020goodness}. However, the stochastic properties, e.g. stationarity and ergodicity, of $\beta$ARMA models are not covered in \cite{scher2020goodness}.

Another compositional time series model is proposed in \cite{zheng2015generalized}. In $\beta$ARMA models, as the link function has the format of an ARMA model, the error term $r_t = g(X_t)-g(\mu_t)$ is expected to be an innovation with expectation 0, which can be used in proving properties of stochastic processes. However, the authors of \cite{zheng2015generalized} point out that in $\beta$ARMA models, the conditional expectation of the error term $r_t$ given the previous information set is not zero according to \cite{zheng2015generalized}. That is, the error sequence is not a martingale difference sequence (MDS). When the error sequence is a MDS, verifiable stationarity and ergodicity conditions can be obtained (\cite{zheng2015generalized}). Therefore, in order to ensure the error sequence is a MDS, in \cite{zheng2015generalized}, the authors extend the original GARMA to martingalized GARMA (M-GARMA) model by using nonidentical transformation function, i.e. $h(.)$, and link function, i.e. $g(.)$, for $X_t$ so that the conditional expectation of the error term, i.e. $h(X_t)-g(\mu_t)$, given the past is zero. They apply it to compositional time series (details will be introduced in Section \ref{sec: Model definition}). The model is referred as Beta-M-GARMA. In \cite{zheng2015generalized}, when the transformation function $h(.)$ is chosen as the logit function, they find the link function $g_{\tau}$ so that the error term, $\text{logit}(X_t)-g_{\tau}(\mu_t)$, is MDS. The stationarity and geometric ergodicity of a time series following this model are proved (Corollary 2) in \cite{zheng2015generalized}. Also, in \cite{zheng2015generalized}, parameter estimation and model selection are discussed as well (Theorem 4). However, the complexity of $g_{\tau}$ causes the infeasibility of getting the closed-form likelihood function, which can be utilized to conduct sequential change-point detection. Following the same idea as \cite{zheng2015generalized}, for multivariate compositional time series, Dirichlet ARMA models are developed in \cite{zheng2017dirichlet}. For compositional time series with long-range dependence, the authors extend the $\beta$ARMA and propose $\beta$ARFIMA (autoregressive fractionally
integrated moving average) model in \cite{pumi2019beta}, equations (4) and (6). For parameter estimation, in \cite{pumi2019beta}, the consistency and asymptotic normality of partial MLE for the parameters are proved (Section 4). Moreover, in \cite{gorgi2021beta}, the authors show the stationarity and ergodicity of a general class of Beta time series subject to a contraction condition. However, the contraction condition is not necessarily guaranteed by the Beta models based on logit link (like \cite{rocha2009beta} and \cite{zheng2015generalized}). The authors discuss Thereshold and STAR Beta AR models, which are the examples of the Beta time series under the contraction condition, in the paper. In the two models, the conditional mean is determined by the previous observation and exogenous variables through a linear function. Thus, in order to guarantee the conditional mean is bounded by 0 and 1, various conservative conditions will be applied to the model depending on how the exogenous variables affect the mean. Considering that, we will mainly focus on Beta time series models based on generalized linear model (GLM). There are many areas with promising applications of Beta regression models, e.g. forecasting mortality rates due to fatal work-related accidents, see \cite{melchior2021forecasting}, nowcasting COVID-19 trends, see \cite{chakraborty2020nowcasting}, and so on.

To the best of our knowledge, the only research that has been done so far on change point detection for compositional time series is \cite{prabuchandran2019change}. In this work, the authors propose a segmentation method locating the optimal point of a series so that the sum of the log-likelihood functions before and after the location reaches the largest. The location is considered as the "change point". Then the change point hypothesis is tested by using likelihood ratio test (LRT). Considering the distinction between retrospective change point detection and sequential change-point detection, we categorize the method as a retrospective method, since the test is based on the whole data set. However, as mentioned in Section \ref{sec: introduction}, sequential change-point detection for compositional time series can be crucial in practice. 

From the review above, we can see that so far there is no model directly built on [0,1] for compositional time series that is proven to be stationary and ergodic, and, in the meantime, has closed-form likelihood function, which is required to develop parametric change-point detection. The related change-point detection research is rare. Therefore, in this paper, we first propose the generalized Beta AR model, that will be shown to have closed-form likelihood function and to enjoy the aforementioned stochastic properties.  Then, we develop a sequential change-point detection method for compositional time series following the model. 

\section{Generalized Beta AR Model based on GLM}
\label{sec: Model definition}
In Section \ref{sec: Literature Review}, we briefly introduce the two existing models based on GLM for compositional time series (\cite{rocha2009beta} and \cite{zheng2015generalized}) and point out the constraints of the two models. Before introducing the generalized Beta AR model, we first compare the two Beta ARMA models developed in \cite{rocha2009beta} and \cite{zheng2015generalized}. We will introduce the two models and analyze the pros and cons of them. Based on the pros and cons, we develop the generalized Beta AR model. \\
\textbf{$\beta$ARMA from \cite{rocha2009beta}} \\
Let $\{X_t\}$ be compositional random variables and denote $\{\bm{W}_t\}$ as a covariate sequence. Given the previous information set $\mathcal{C}_{t} = \sigma(\bm{W}_t,X_{t-1},\bm{W}_{t-1},...,\bm{W}_1,X_0)$, $X_t|\mathcal{C}_{t} \sim \text{Beta}(\tau \mu_t, \tau (1-\mu_t))$. That is,
\begin{equation}
\label{eq: conditional density function of X_t given C_t}
    f(X_t|\mathcal{C}_{t}) = \frac{1}{B(\tau \mu_t, \tau (1-\mu_t))}X_t^{\tau \mu_t-1}(1-X_t)^{\tau(1-\mu_t)-1}, 0<X_t<1,
\end{equation}
where $\mathbbm{E}(X_t|\mathcal{C}_t) = \mu_t$ and $Var(X_t|\mathcal{C}_t) = \mu_t(1-\mu_t)/(\tau+1).$
The conditional expectation $\mu_t$ is determined through
\begin{equation}
\label{eq: link function rocha 2009}
    g(\mu_t) = \varphi_0+ \bm{W}_t^T \bm{\phi}+ \sum_{i=1}^p \varphi_i \{g(X_{t-i})-\bm{W}_{t-i}^T \bm{\phi}\}+\sum_{j=1}^q \theta_j r_{t-j},
\end{equation}
where $r_{t} = g(X_t)-g(\mu_t).$
This model incorporates past observations of the series $\{X_t\}$ and the corresponding covariates $\{\bm{W}_t\}$. The conditional expectation of current time, $\mu_t$, is determined by past information through a twice differentiable strictly monotonic link function $g: (0,1) \rightarrow \mathbbm{R}$. Some common choices of $g$ are logit, probit, and complementary log-log link functions. In \cite{rocha2009beta}, the authors choose the logit link and fit $\beta$ARMA models to the rate of hidden unemployment in Brazil. From (\ref{eq: link function rocha 2009}) we can see that it uses a link function to be the same as the transformation function for $X_t$ (referred as $x$-link function). However, the conditional expectation of the error term, $\mathbbm{E}(g(X_t)-g(\mu_t)|\mathcal{C}_t)$, is not zero, which makes stationarity and geometric ergodicity of the process difficult to establish, according to \cite{zheng2015generalized}. Considering that, the following model is built. \\

\textbf{Logit-Beta-M-ARMA from \cite{zheng2015generalized}} \\
\begin{equation*}
X_t|\mathcal{C}_{t} \sim \text{Beta}(\tau \mu_t, \tau (1-\mu_t)),  0<X_t<1,
\end{equation*}
and 
\begin{equation*}
    g_{\tau}(\mu_t) = \varphi_0+\sum_{i=1}^p \varphi_i  \text{logit}(X_{t-i})+ \sum_{j=1}^q \theta_j r_{t-j},
\end{equation*}
where $g_{\tau}(\mu_t) = \psi(\tau \mu_t)-\psi(\tau (1-\mu_t))$, of which $\psi$ is the digamma function, and $r_t = \text{logit}(X_t) - g_{\tau}(\mu_t)$. A detailed introduction to the digamma function $\psi(\cdot)$ and polygamma functions, the sequence of derivatives of $\psi$, can be found in \cite{batir2007some} (the series representations are shown in (1.1) and (1.2)). By using the link function $g_{\tau}$, the error term $\text{logit}(X_t)-g_{\tau}(\mu_t)$ is proven as a MDS in \cite{zheng2015generalized}. Relying on the fact that the error sequence is MDS, the stationarity and the geometric ergodicity are proven in \cite{zheng2015generalized}. However, because of the complexity of $g_{\tau}$, it is hard to get the closed-form $g^{-1}_{\tau}$, which will be used in the future likelihood estimation. Also, the model doesn't include exogenous variables. In order to integrate the advantages from both models, we construct the following model. 

\textbf{Generalized Beta AR($p$)}\\
From the two Beta ARMA models introduced above, we can see that the Beta distribution can be used to model compositional data and GLM is an adoptable approach to model time-varying expectation of the Beta random variables. However, as discussed before, the two models have their drawbacks. In order to develop sequential change-point detection methods, we need to derive asymptotic behavior of test statistics of the process. It is helpful if the process is strictly stationary and geometrically ergodic. But the properties are not shown for $\beta$ARMA. Also, as $X_{t-i}$ approaches to 0 or 1, $g(X_{t-i})$ approaches $-\infty$ or $\infty$. With nonzero $\varphi_i$, according to (\ref{eq: link function rocha 2009}), $g(\mu_t)$ will also approach $-\infty$ or $\infty$. It follows that $\mu_t$ will be trapped at 0 or 1, which makes $X_t$ generated from the model a degenerate random variable at 0 or 1. In practice, the unbounded $g$ would cause overflow errors when $X_t$ is close to 0 or 1.

Logit-Beta-M-ARMA has been shown strictly stationary and geometrically ergodic in \cite{zheng2015generalized}. However, the complexity of the inverse of the link function $g_{\tau}$ makes parameter estimation and challenging. In Remark 2.6 of \cite{zheng2015generalized}, the authors propose to use second-order Taylor expansion to approximate $g_{\tau}$, and then estimate $g_{\tau}^{-1}$ numerically. It would bring variance to the likelihood function, which will affect the parameter estimation accuracy. Also, it would make the following change point detection method hard to implement, since it relies on $g_{\tau}^{-1}.$ 

Motivated by the drawbacks of the previous models, we build the following generalized Beta AR($p$) model. 
\begin{equation}
\label{equation: model definition}
    X_t| \mathcal{C}_t = X_t|\mathbbm{Z}_{t-1}\sim \text{Beta}(\tau \mu_t, \tau(1-\mu_t)),0
    \leq X_t \leq 1,
\end{equation}
where
\begin{equation}
\label{eq: link function of GBETAAR}
g(\mu_t) =  \bm{\beta}^T \mathbbm{Z}_{t-1} = \varphi_0+ \sum_{i=1}^p \varphi_i \mathcal{A}(X_{t-i})+\bm{W}_{t}^T \bm{\phi},
\end{equation}
in which 
\begin{equation}
\label{eq: notations in the definition of generalized beta}
    \begin{aligned}
    & \mathcal{C}_{t} = \sigma(\bm{W}_t,X_{t-1},\bm{W}_{t-1},...,\bm{W}_1,X_0),\\
    & g(x) =\log(x/(1-x)),\\
    & \mathbbm{Z}_{t-1} = (1, \mathcal{A}(X_{t-1}),...,\mathcal{A}(X_{t-p}),\bm{W}_{t}).
    \end{aligned}
\end{equation}
The parameter vector $\bm{\beta} = (\varphi_0,\varphi_1,...,\varphi_p,\bm{\phi})^T \in \mathbbm{R}^{1+l+p}$ and the dispersion parameter $\tau>0$. In general, the link function $g$ can be any twice differentiable strictly monotonic function. Considering the complexity of getting the close-formed $g^{-1}$, which will be used in parameter estimation, we choose the logit link. Instead of mapping the unit interval to $\mathbbm{R}$, the $x$-link function $\mathcal{A}$ maps the unit interval $[0,1]$ to a bounded interval $[L,U]$. The choices of $\mathcal{A}$ rely on how the past observations affect the conditional expectation. For example, when $\mathcal{A}$ is an identity function, the regressor $\mathbbm{Z}_{t-1}$ is similar with the regressor in the Binomial AR($p$) model in \cite{Liu_2025}. If $\mathcal{A}(x) = \text{logit}(x^*)$ where $x^* = \min\{\max(c,x),1-c\}, 0<c<\frac{1}{2}$, the transformation coincides with Case 3 proposed in \cite{woodard2010stationarity}. When the threshold $c$ is chosen to be close to 0, $\mathcal{A}(x)$ is equivalent to $\text{logit}(x)$ except when $x$ approaches to 0 or 1. Therefore, it can inherit most structures and features of the $\beta$ARMA model. In the meantime, thanks to the boundness of the $x$-link function $\mathcal{A}(x)$, generalized Beta AR($p$) allows outputs to be exact 0 or 1 without causing $\mathcal{A}(x)$ to go to $-\infty$ or $\infty$. That is, it conquers the aforementioned constraint of both $\beta$ARMA model and Logit-M-Beta-ARMA model. The boundness of $\mathcal{A}(x)$ would be beneficial for the proofs of strict stationarity and geometric ergodicity of the process. The details will be shown in Sec \ref{sec: Properties of Generalized Beta AR($1$)}. 

Compared with models with a fixed $\mathcal{A}(x)$, e.g. \cite{rocha2009beta}, generalized Beta AR model allows flexible choices of $\mathcal{A}(x)$. Based on a real-life application, in Section \ref{sec: real life data}, the forecasting abilities of models with different $\mathcal{A}(x)$ will be discussed.
\subsection{Properties of Beta AR($1$)}
\label{sec: Properties of Generalized Beta AR($1$)}

For the sake of simplicity, we refer to the generalized Beta AR($p$) model as Beta AR($p$) model onward. From (\ref{eq: link function of GBETAAR}, we can see that so far $\{X_t\}$ is not a Markov chain, as it depends on not only the external variable $\bm{W}_t$ but also the $p$ previous observations $X_{t-1}, X_{t-2},...,X_{t-p}$. In order to take advantage of the Markov chain-based theorems to prove the stochastic properties (e.g. strict stationarity and geometric ergodicity) of $\{X_t\}$, we focus on the case when $p=1$ in the rest of the chapter. As mentioned before, the distribution of $\{X_t\}$ replies on both $\bm{W}_t$ and $X_{t-1}$. So we first propose some regularity conditions of $\bm{W}_t$ for the further Markovian of $\{X_t\}$.

\begin{assumption}
\label{assumption: properties of W_t}
The $l$-dimensional covariate vector process $\{\bm{W}_t\}$ is a strict stationary and geometrically ergodic Markov chain defined on a bounded state space $[a,b]^l$, where $a$ and $b$ are known. Given $\bm{W}_{t-1}$, we denote the transition density function as $h_{w_{t-1}}(\bm{W}_{t}):=h(\bm{W}_{t}|\bm{W}_{t-1} = w_{t-1}).$ $h$ is assumed to be known and continuous with respect to $\bm{W}_{t-1}.$ 
\end{assumption}
Define $\bm{Y}_{t} = (X_t,\bm{W}_t)^T \in [0,1] \times [a,b]^l:= \Omega_X \times \Omega_W:= \Omega$. Under Assumption \ref{assumption: properties of W_t}, we can see that the distribution of $\bm{Y}_t$ can be determined once $\bm{Y}_{t-1}$ is given. That is, $\bm{Y}_t$ becomes the Markovian realization of $\{X_t\}$. 

\begin{equation*}
    \bm{Y}_t | \bm{Y}_{t-1} = (X_t|\bm{W}_t, X_{t-1}) \times (\bm{W}_t|\bm{W}_{t-1}).
\end{equation*}

We will investigate the Markov chain properties of $\{\bm{Y}_{t}\}$ in this subsection. 

We first establish the $\psi$-irreducibility and the aperiodicity of $\{\bm{Y}_{t}\}$. $\psi$-irreducibility and aperiodicity of a process are fundamental for proving the stationarity and geometric ergodicity of the process. $\psi$-irreducibility is used to describe that any subset of a chain is reachable eventually regardless of the starting point, if the maximal irreducibility measure $\psi$ on the subset is not zero. Aperiodicity depicts a cyclical behavior of the process. In a nutshell, the process can reach any state of the state space at $t+1$ regardless of the value of the state at $t$. Then thanks to the continuities of the Beta distribution and the conditional density function of the covariate vector $h_w(\cdot)$, we can first prove $\{\bm{Y}_t\}$ is a weak Feller chain. As a weak Feller chain, the state space of $\{\bm{Y}_t\}$, $\Omega$, is a small set. By taking advantage of these properties, the strict stationarity and geometric ergodicity can be proven. 
\begin{lemma}
\label{lemma:irreducibility&aperiodicity}
Under Assumption \ref{assumption: properties of W_t}, the Markovian realization process $\{\bm{Y}_t\}$ is a $\psi$-irreducible and aperiodic Markov chain. 
\end{lemma}
\begin{proof}
As shown in Proof of Theorem 3 in \cite{woodard2010stationarity}, aperiodicity and $\varphi$-irreducibility are immediate since the transition kernel is positive on $\Omega$. Then, by Theorem 4.0.1 and Prop 4.2.2 in \cite{meyn1993markov}, there exists an essentially unique maximal irreducibility measure $\psi$ on $\mathcal{B}(\Omega)$ that the process is $\psi$-irreducible.
\end{proof}
\begin{lemma}
\label{lemma:weak Feller}
$\{\bm{Y}_t\}$ is a weak Feller chain and the state space $\Omega_X \times \Omega_W$ is a small set. 
\end{lemma}
The proof can be found in Section \ref{app: weak feller chain}. 

Finding a small set and then taking advantage of the drift condition (details can be found in the proof of Theorem \ref{theorem:stationary&ergodic}) to prove the geometric ergodicity of a process is a typical procedure. Unlike some finite space Markov chains, the existence of a small set is not very obvious for general state space chains. Therefore, we first prove $\{\bm{Y}_t\}$ is a Feller chain. It can further guarantee that the state space $\Omega$ is a small set. 

Given that the state space is a small set, we can first take advantage of the drift condition to prove the geometric ergodicity of the process. The existence of a unconditional probability follows. Then, by using the nature of the model, we can show that the process is strictly stationary. 

\begin{theorem}
\label{theorem:stationary&ergodic}
The time series $\{\bm{Y}_t\}$ is a strictly stationary and geometrically ergodic process.
\end{theorem}
The proof can be found in Section \ref{app: stationary ergodicity}. 
\subsection{Parameter estimation}
\label{subsec: parameter estimation}
In this subsection, we will discuss parameter estimation for the true parameter vector $\bm{\beta}_0$ for the link function and the true dispersion parameter $\tau_0$ by using partial maximum likelihood estimation (PMLE). In order to prove the consistency and asymptotic normality of the PMLE, we get the closed-form partial log-likelihood function of the Beta AR(1) model. We prove the score vector is a martingale difference sequence (MDS), which will be useful in building the test statistic for change-point detection in Section \ref{sec: monitoring scheme and null}. 

\subsubsection{Partial maximum likelihood estimation (PMLE)}

\begin{assumption}
\label{assumption: compact}
The parameter vector $\bm{\eta}:=(\tau,\bm{\beta})$ is defined on a convex and compact set $ \Xi \subset (0,\infty) \times \mathbbm{R}^{l+2}$, where $l$ is the dimension of the exogenous variable vector $\bm{W}_t$, and the true parameter $\bm{\eta}_0 = (\tau_0, \bm{\beta}_0)$ is in the interior of $\Xi$. 
\end{assumption}

Assumption \ref{assumption: compact} is a general assumption for the parameter space. It, together with the boundness of $\{\bm{Y}_t\}$, guarantees the availability of PMLE. The details will be discussed later in this section. 

Denote the set with the first $m+1$ observations as $\bm{Y} = (X_0,\bm{W}_0,X_1,\bm{W}_1,...,X_m,\bm{W}_m)$. From Theorem \ref{theorem:stationary&ergodic} we know that the process is strictly stationary. With parameter $\bm{\eta}$, we denote the unconditional pdf of the first observation $\bm{Y}_0 = (X_0,\bm{W}_0)$ as $\pi_{\bm{\eta}}(\bm{Y}_0).$ Then the likelihood function based on the first $m+1$ observations is 
\begin{equation*}
    L_{\bm{\eta}}(\bm{Y}) = \pi_{\bm{\eta}}(\bm{Y}_0) \Big[\prod_{t=1}^m R_{\bm{\eta}}(\bm{Y}_t| \mathcal{F}_{t-1}) \Big],
\end{equation*}
 where $\mathcal{F}_{t-1} = \sigma( X_{t-1},\bm{W}_{t-1},...,\bm{W}_0, X_0)$ and $R_{\bm{\eta}}(\bm{Y}_t| \mathcal{F}_{t-1})$ is the conditional density function of $\bm{Y}_{t} = (X_t,\bm{W}_t)$ given past information $\mathcal{F}_{t-1}$.
 Considering that at time $t$, we first observe $\bm{W}_{t}$, and $X_t$ follows as a Beta random variable with mean determined by $X_{t-1}$ and $\bm{W}_t$, $R_{\bm{\eta}}(\bm{Y}_t|\mathcal{F}_{t-1})$ can be further decomposed as 
 \begin{equation*}
 \begin{aligned}
       R_{\bm{\eta}}(\bm{Y}_t|\mathcal{F}_{t-1}) &= h(\bm{W}_t|\mathcal{F}_{t-1}) f_{\bm{\eta}}(X_t|\mathcal{C}_{t}) \\
       & = h_{w_{t-1}}(\bm{W}_{t})f_{\bm{\eta}}(X_t|\mathcal{C}_t)
 \end{aligned}
 \end{equation*}
where 
\begin{equation}
\label{eq: beta ar conditional density function of X_t given C_t}
    f_{\bm{\eta}}(X_t|\mathcal{C}_{t}) = \frac{1}{B(\tau \mu_t, \tau (1-\mu_t))}X_t^{\tau \mu_t-1}(1-X_t)^{\tau(1-\mu_t)-1}, 0\leq X_t\leq1,
\end{equation}
with $\mu_t$ defined as (\ref{eq: link function of GBETAAR}) and $\mathcal{C}_t$ is the information set defined in (\ref{eq: notations in the definition of generalized beta}).  

In \cite{Liu_2025}, the authors take advantage of the partial likelihood function and prove the consistency and asymptotic normality of PMLE. Similarly, we can define the partial likelihood function for the generalized Beta AR(1) model $$PL_{m}(\bm{\eta}):= \prod_{t=1}^m f_{\bm{\eta}}(X_t|\mathcal{C}_{t}).$$

The PMLE, $\hat{\bm{\eta}} := (\hat{\tau},\hat{\bm{\beta}})$, can be drawn from
\begin{equation}
\label{eq: PMLE objective function}
    \underset{\bm{\eta} \in \Xi}{\arg\max} \log PL_{m}(\bm{\eta}).
\end{equation}

\subsubsection{Closed-form expressions of the partial log-likelihood function, the score vector and the Hessian matrix}
\label{subsec: Closed-form expressions of the partial log-likelihood function, the score vector and the Hessian matrix}
As discussed in Section \ref{sec: Properties of Generalized Beta AR($1$)}, one of the motivations of developing generalized Beta AR models is to have closed-form partial likelihood function and closed-form score vector. In this part, we provide the closed-form expressions of the partial log-likelihood function, the score vector and the Hessian matrix. The score vector and the Hessian matrix will be used in proving the consistency and the asymptotic normality of $\hat{\bm{\eta}}$. The score vector will be further used in building the test statistic for change-point detection.

According to (\ref{equation: model definition}), (\ref{eq: link function of GBETAAR}) and (\ref{eq: beta ar conditional density function of X_t given C_t}), the partial log-likelihood function
\begin{equation}
\label{eq: closed-form of the partial log-likelihood function}
    \begin{aligned}
        \log PL_{m}(\bm{\eta})&= \sum_{t=1}^m \log f_{\bm{\eta}}(X_t|\mathcal{C}_t) \\
        &= \sum_{t=1}^m \log f_{\bm{\eta}}(X_t|X_{t-1},\bm{W}_t) \\
        &= \sum_{t=1}^m \Big(\log\Gamma(\tau) -\log \Gamma(\tau \mu_t)-\log\Gamma(\tau(1-\mu_t)) \\
        & \quad \quad \quad +(\tau \mu_t-1) \log X_t+(\tau(1-\mu_t)-1) \log (1-X_t) \Big).
    \end{aligned}
\end{equation}
The first derivative of $\log PL_m(\bm{\eta})$ with respect to $\tau$ is
\begin{equation*}
\begin{aligned}
    \frac{\partial \log PL_{m}(\bm{\eta})}{\partial \tau} & = \sum_{t=1}^m \Big(\psi(\tau)-\mu_t \psi(\tau \mu_t)-(1-\mu_t) \psi(\tau(1-\mu_t))+\mu_t \log X_t+(1-\mu_t) \log(1-X_t) \Big)\\
    &= \sum_{t=1}^m \Big(\mu_t(X_t^*-\mu_t^*)+\log(1-X_t)-\psi(\tau(1-\mu_t))+\psi(\tau) \Big),
\end{aligned}
\end{equation*}
where $X_t^* = \log(\frac{X_t}{1-X_t})$, $\mu_t^* = \psi(\tau \mu_t)-\psi(\tau(1-\mu_t))$. The digamma function $\psi(\cdot)$ is defined as the logarithmic derivative of the gamma function.
\begin{equation*}
    \psi(x) = \frac{\text{d}\log \Gamma(x)}{\text{d}x} = \frac{\Gamma' (x)}{\Gamma(x)}.
\end{equation*}
For the $i$th element of the parameter vector $\bm{\beta}$, $\beta_i$, denote the corresponding regressor as $\mathbbm{Z}_{t,i}$, then the first derivative with respect to $\beta_i$ is
\begin{equation*}
    \begin{aligned}
\frac{\partial \log PL_{m}(\bm{\eta})}{\partial \beta_i} & = \sum_{t=1}^m \frac{\partial \log f_{\bm{\eta}}(X_t|X_{t-1},\bm{W}_t)}{\partial \mu_t} \frac{\partial \mu_t}{\partial \beta_i} \\
&= \sum_{t=1}^m \Big(-\tau \frac{\Gamma'(\tau \mu_t)}{\Gamma(\tau \mu_t)}+\tau \frac{\Gamma'(\tau(1-\mu_t))}{\Gamma(\tau(1-\mu_t))}+\tau \log X_t -\tau \log(1-X_t) \Big)\frac{\partial \mu_t}{\partial \beta_i} \\
&= \sum_{t=1}^m  \tau\Big(\log(\frac{X_t}{1-X_t})-(\psi(\tau \mu_t)-\psi(\tau(1-\mu_t)))\Big)\frac{\partial \mu_t}{\partial \beta_i}\\
&= \sum_{t=1}^m  \tau(X_t^*-\mu_t^*) \frac{\partial \mu_t}{\partial \beta_i} \\
&= \sum_{t=1}^m \tau(X_t^*-\mu_t^*) \frac{\partial g^{-1}(\mathbbm{Z}_{t-1}^T \bm{\beta})}{\partial \beta_i} \\
&= \sum_{t=1}^m \tau(X_t^*-\mu_t^*) \frac{\exp(-\mathbbm{Z}_{t-1}^T \bm{\beta})}{(1+\exp(-\mathbbm{Z}_{t-1}^T \bm{\beta}))^2} \mathbbm{Z}_{t-1,i} \\
&= \sum_{t=1}^m \tau(X_t^*-\mu_t^*)\mu_t(1-\mu_t) \mathbbm{Z}_{t-1,i}.
\end{aligned}
\end{equation*}

Denote the score vector $ \triangledown \log PL_{m}(\bm{\eta}) \in \mathbbm{R}^{l+3}$ as $S_m(\bm{\eta})$. Based on the calculation above, the closed-form expression of $S_m(\bm{\eta})$ is
\small
\begin{equation}
\label{eq: score vector}
    S_m(\bm{\eta}) = \sum_{t=1}^m \begin{bmatrix} \mu_t(X_t^*-\mu_t^*)+\log(1-X_t)-\psi(\tau(1-\mu_t))+\psi(\tau) \\
    \tau(X_t^*-\mu_t^*)\mu_t(1-\mu_t) \mathbbm{Z}_{t-1}
    \end{bmatrix} \overset{\Delta}{=} \sum_{t=1}^m G(X_t,\bm{\eta}|
    \mathcal{C}_t).
\end{equation}
\normalsize

Denote the $i$th element of the score vector $S_m(\bm{\eta})$ as $S_m(\bm{\eta})_i$, then the Hessian matrix 
\begin{equation*}
    \triangledown S_m(\bm{\eta}) = \begin{bmatrix} \frac{\partial S_m(\bm{\eta})_1}{\partial \tau} & \frac{\partial S_m(\bm{\eta})_1}{\partial \beta_1}& \cdots & \frac{\partial S_m(\bm{\eta})_1}{\partial \beta_{l+2}} \\
    \frac{\partial S_m(\bm{\eta})_2}{\partial \tau} & \frac{\partial S_m(\bm{\eta})_2}{\partial \beta_1} & \cdots & \frac{\partial S_m(\bm{\eta})_2}{\partial \beta_{l+2}} \\
    \vdots & \vdots & \ddots & \vdots \\
    \frac{\partial S_m(\bm{\eta})_{l+3}}{\partial \tau} & \frac{\partial S_m(\bm{\eta})_{l+3}}{\partial \beta_1}& \cdots & \frac{\partial S_m(\bm{\eta})_{l+3}}{\partial \beta_{l+2}}
\end{bmatrix}.
\end{equation*}
Denote the $(i,j)$th element of $\triangledown S_m(\bm{\eta})$ as $[\triangledown S_m(\bm{\eta})]_{i,j}$. 
\begin{equation*}
    \begin{aligned}
        [\triangledown S_m(\bm{\eta})]_{1,1} =& \frac{\partial S_m(\bm{\eta})_1}{\partial \tau} \\
        =& \sum_{t=1}^m \frac{\partial}{\partial \tau} \Big(\mu_t(X_t^*-\mu_t^*)+\log(1-X_t)-\psi(\tau(1-\mu_t))+\psi(\tau)\Big)\\
        =& \Big(\sum_{t=1}^m -\mu_t \frac{\partial \mu_t^*}{\partial \tau}-\frac{\partial}{\partial \tau} \psi(\tau(1-\mu_t))+\frac{\partial}{\partial \tau}\psi(\tau) \Big)\\
        =& \sum_{t=1}^m \psi'(\tau)-\mu_t^2 \psi'(\tau \mu_t)-(1-\mu_t)^2 \psi'(\tau(1-\mu_t)).
    \end{aligned}
\end{equation*}
For $i = 1,2,...,l+2$,
\begin{equation*}
    \begin{aligned}
        [\triangledown S_m(\bm{\eta})]_{i+1,1} =& \frac{\partial S_m(\bm{\eta})_{i+1}}{\partial \tau} \\
        =& \sum_{t=1}^m \frac{\partial}{\partial \tau}\tau(X_t^*-\mu_t^*)\mu_t(1-\mu_t) \mathbbm{Z}_{t-1,i} \\
        =& \sum_{t=1}^m [(X_t^*-\mu_t^*)-\tau\frac{\partial \mu_t^*}{\partial \tau}]\mu_t(1-\mu_t) \mathbbm{Z}_{t-1,i} \\
        =& \sum_{t=1}^m [(X_t^*-\mu_t^*)-\tau(\mu_t\psi'(\tau \mu_t)-(1-\mu_t)\psi'(\tau(1-\mu_t))]\mu_t(1-\mu_t) \mathbbm{Z}_{t-1,i}. \\
\end{aligned}
\end{equation*}
\begin{equation*}
    \begin{aligned}
        [\triangledown S_m(\bm{\eta})]_{1,i+1} =& \frac{\partial S_m(\bm{\eta})_{1}}{\partial \beta_i} \\
        =& \sum_{t=1}^m \frac{\partial G(X_t,\bm{\eta})_{1}}{\partial \mu_t}\frac{\partial \mu_t}{\partial \beta_{i}} \\
        =& \sum_{t=1}^m [(X_t^*-\mu_t^*)-\mu_t(\tau \psi'(\tau \mu_t)+\tau \psi'(\tau (1-\mu_t)))+\tau \psi'(\tau (1-\mu_t))]\mu_t(1-\mu_t)\mathbbm{Z}_{t-1,i} \\
        =& \sum_{t=1}^m [(X_t^*-\mu_t^*)-\tau(\mu_t\psi'(\tau \mu_t)-(1-\mu_t)\psi'(\tau(1-\mu_t))]\mu_t(1-\mu_t) \mathbbm{Z}_{t-1,i}. \\
        [\triangledown S_m(\bm{\eta})]_{1,i+1}=&[\triangledown S_m(\bm{\eta})]_{i+1,1}.\end{aligned}
\end{equation*}
For $i = 1,2,...,l+2$ and $j = 1,2,...,l+2$,
\begin{equation*}
    \begin{aligned}
        [\triangledown S_m(\bm{\eta})]_{i+1,j+1} = & \frac{\partial S_m(\bm{\eta})_{i+1}}{\partial \beta_{j+1}} \\
        =& \sum_{t=1}^m \frac{\partial G(X_t,\bm{\eta})_{i+1}}{\partial \mu_t} \frac{\partial \mu_t}{\partial \beta_{j+1}} \\
        =& \sum_{t=1}^m \tau \mathbbm{Z}_{t-1,i}[-\frac{\partial \mu_t^*}{\partial \mu_t}\mu_t(1-\mu_t)+(X_t^*-\mu_t^*)\frac{\partial \mu_t(1-\mu_t)}{\partial \mu_t})]\frac{\partial \mu_t}{\partial \beta_{j+1}}\\
        =& \sum_{t=1}^m \tau [-\tau(\psi'(\tau \mu_t)+\psi'(\tau(1-\mu_t)))\mu_t(1-\mu_t)+(1-2\mu_t)(X_t^*-\mu_t^*)]\mu_t(1-\mu_t)\mathbbm{Z}_{t-1,i}\mathbbm{Z}_{t-1,j}.
        \end{aligned}
\end{equation*}
\begin{remark}
\label{remark: continuities of the closed-forms}
$pl_m(\bm{\eta}), S_m(\bm{\eta})$ and $\triangledown S_m(\bm{\eta})$ are all continuous with respect to $\bm{\eta} \in \Xi$. 
It is easy to verify since the digamma function $\psi$ and polygamma function $\psi'$ are continuous with respect to $\mu_t$ and $\tau$ (the properties of polygamma functions can be found in \cite{batir2007some}) and $\mu_t$ is a continuous function with respect to $\bm{\eta}$. 
\end{remark}
\subsubsection{Consistency and asymptotic normality of $\hat{\bm{\eta}}$}
In this part, we will prove the consistency and asymptotic normality of the PMLE $\hat{\bm{\eta}}$. The properties of the score vector will also be discussed. 

Denote the expectation of the conditional log-likelihood function $\mathbbm{E}(\log f_{\bm{\eta}}(X_t|\mathcal{C}_t))$ as $pl(\bm{\eta})$. Denote the rescaled partial log-likelihood function $\frac{1}{m} \log PL_m(\bm{\eta})$ as $pl_m(\bm{\eta})$. In order to investigate the properties of PMLE $\hat{\bm{\eta}}$, which is the maximizer of the conditional log-likelihood function $\log PL_m(\bm{\eta})$, we first show that the true parameter $\bm{\eta}_0$ is the unique maximizer of $pl(\bm{\eta})$. The ergodicity of $pl_m(\bm{\eta})$ will be shown. They will be used in the proof of the consistency of $\hat{\bm{\eta}}$. 

\begin{lemma}
\label{lemma: unique of eta_0}
The true parameter $\bm{\eta}_0$ is the unique maximizer of $pl(\bm{\eta})$ in the parameter space $\Xi$. That is, 
\begin{equation*}
    pl(\bm{\eta}_0)>pl(\bm{\eta}), \text{ for any } \bm{\eta} \in \Xi \text{ and } \bm{\eta} \neq \bm{\eta}_0. 
\end{equation*}
\end{lemma}
The detailed proof can be found in Section \ref{app: unique of eta_0}.

\begin{lemma}
\label{lemma: ergodicity of log-partial likelihood}
The rescaled partial log-likelihood function
\begin{equation*}
pl_m(\bm{\eta}) = \frac{1}{m}\log PL_m(\bm{\eta}) = \frac{1}{m} \sum_{t=1}^m \log f_{\bm{\eta}}(X_t|\mathcal{C}_{t}) \overset{a.s.}{\rightarrow} pl(\bm{\eta}) \text{ as }m \rightarrow \infty
\end{equation*}
for any $\bm{\eta} \in \Xi$.
\end{lemma}

\begin{proof}
According to Theorem \ref{theorem:stationary&ergodic}, the Markov chain $\bm{Y}_t = (X_t, \bm{W}_t)$ is geometrically ergodic. From the expression shown in (\ref{eq: closed-form of the partial log-likelihood function}), it's easy to conclude the continuity of $\log f_{\bm{\eta}}(X_t|\mathcal{C}_t)$ with respect to $\bm{Y}_t$. Based on Proposition 2.1.1 and Section 2.2 of \cite{Strauman2005Daniel}, $\log f_{\bm{\eta}}(X_t|\mathcal{C}_t)$, as a continuous function of $\bm{Y}_t$, is ergodic. Applying the ergodic theorem,
\begin{equation*}
pl_m(\bm{\eta})\overset{a.s.}{\rightarrow} \mathbbm{E}(\log f_{\bm{\eta}}(X_t|\mathcal{C}_t)) = pl(\bm{\eta}) \text{, as }m \rightarrow \infty.
\end{equation*}
\end{proof}

According to Lemma \ref{lemma: unique of eta_0}, we can find a compact neighborhood around $\bm{\eta}_0$ so that $pl(\bm{\eta})$ is concave in the neighborhood. Combining the concavity of $pl(\bm{\eta})$ in the neighborhood and the pointwise convergence of $pl_m(\bm{\eta})$ stated in Lemma \ref{lemma: ergodicity of log-partial likelihood}, we can conclude that as $m\rightarrow \infty$, $pl_m(\bm{\eta})$ is a concave function in the compact neighborhood around $\bm{\eta}_0$. Therefore, the PMLE $\hat{\bm{\eta}}$ that optimizes (\ref{eq: PMLE objective function}) has
\begin{equation}
\label{eq: Sm=0}
    S_m(\hat{\bm{\eta}}) = 0.
\end{equation}

\begin{theorem}
\label{theorem: consistency}
The PMLE $\hat{\bm{\eta}}$ is strongly consistent.
\begin{equation*}
    \hat{\bm{\eta}} \overset{a.s.}{\rightarrow}\bm{\eta}_0 \text{ , as } m \rightarrow \infty.
\end{equation*}
\end{theorem}
\begin{proof}
According to Lemma \ref{lemma: unique of eta_0}, the true parameter $\bm{\eta}_0$ is the unique maximizer of $pl(\bm{\eta})$. We also point out the almost sure convergence of the average log-partial likelihood function, $pl_m(\bm{\eta}) \overset{a.s.}{\rightarrow} pl(\bm{\eta})$, in Lemma \ref{lemma: ergodicity of log-partial likelihood}. Therefore, similar with the proof of Theorem 4.1 of \cite{straumann2006quasi}, according to \cite{wald1949note}, the strong consistency of $\hat{\bm{\eta}}$ can be shown.
\begin{equation*}
    \hat{\bm{\eta}} \overset{a.s.}{\rightarrow}\bm{\eta}_0.
\end{equation*}
\end{proof}

Before proving the asymptotic normality of $\hat{\bm{\eta}}$, we first show that $G(X_t,\bm{\eta}_0|\mathcal{C}_t)$ is a martingale difference sequence (MDS), and naturally, the sum of $G$, $S_m(\bm{\eta}_0)$, is a martingale with respect to the filter $\mathcal{C}_t$. Relying on the martingale $S_m(\bm{\eta}_0)$, we would like to use martingale central limit theorem (MCLT) to prove the asymptotic normality of $\hat{\bm{\eta}}$.

\begin{lemma}
\label{lemma: G is MDS}
$G(X_t, \bm{\eta}_0|\mathcal{C}_t)$ is a martingale difference sequence (MDS).
\end{lemma}
\begin{proof}
In order to calculate the expectation of $G$, we need to calculate the expectation of $X_t^* = \log (\frac{X_t}{1-X_t})$, the transformed random variable of the Beta random variable $X_t \sim$Beta $(\tau \mu_t, \tau(1-\mu_t)).$

For a random variable $X \sim $Beta $(\alpha,\beta)$, the moment generating function for $\log(\frac{X}{1-X})$ is 
\begin{equation*}
\begin{aligned}
    M_{\log(\frac{X}{1-X})}(t) &= \mathbbm{E}(e^{t\log(\frac{X}{1-X})}) \\
    &=\int_0^1\frac{\Gamma(\alpha+\beta)}{\Gamma(\alpha)\Gamma(\beta)}e^{t\log(\frac{x}{1-x})}x^{\alpha-1}(1-x)^{\beta-1} \text{d}x \\
        &= \frac{\Gamma(\alpha+\beta)}{\Gamma(\alpha)\Gamma(\beta)}\int_0^1 x^{\alpha+t-1}(1-x)^{\beta-t-1} \text{d}x \\
        &= \frac{\Gamma(\alpha+\beta)}{\Gamma(\alpha)\Gamma(\beta)} \frac{\Gamma(\alpha+t)\Gamma(\beta-t)}{\Gamma(\alpha+\beta)} \\
        &=\frac{\Gamma(\alpha+t)\Gamma(\beta-t)}{\Gamma(\alpha)\Gamma(\beta)}, \;|t|<\min\{\alpha,\beta\}.
\end{aligned}
\end{equation*}
Then the cumulant-generating function is 
\begin{equation}
\label{eq: cgf for logX/(1-X)}
    K_{\log(\frac{X}{1-X})}(t) = \log M_{\log(\frac{X}{1-X})}(t) =\log \Gamma(\alpha+t)+\log \Gamma(\beta-t)-\log(\Gamma(\alpha)\Gamma(\beta)).
\end{equation}
The expectation of $X^*$ can be derived from taking the first derivative of (\ref{eq: cgf for logX/(1-X)}) at $t=0$.
\begin{equation*}
    \mathbbm{E}(X^*) = \Big (\frac{\text{d}\log \Gamma(\alpha+t)}{\text{d}t}+ \frac{\text{d}\log \Gamma(\beta-t)}{\text{d}t} \Big) |_{t = 0} = \frac{\Gamma'(\alpha)}{\Gamma(\alpha)}-\frac{\Gamma'(\beta)}{\Gamma(\beta)} \overset{\Delta}{=} \psi(\alpha)-\psi(\beta).
\end{equation*}
So $\mathbbm{E}(X^*_t) = \psi(\tau \mu_t)-\psi(\tau (1-\mu_t)),$ which, according to the notation used in (\ref{eq: score vector}), is $\mu^*_t$. Similar with the method we use to calculate the expectation of $X^*_t$, we are able to get the cumulant-generating function of $\log(1-X)$ and conclude that
\begin{equation*}
    \mathbbm{E}(\log(1-X_t)) = \psi(\tau(1-\mu_t))-\psi(\tau).
\end{equation*}
As a statistic built on $\mathcal{C}_t$ (so $\mu_t$ is known), the expectation of the first element of $G(X_t, \bm{\eta}_0|\mathcal{C}_t)$,
\begin{equation*}
\begin{aligned}
    \mathbbm{E}\Big(G(X_t, \bm{\eta}_0|\mathcal{C}_t)_1\Big) &= \mathbbm{E}\Big(\mu_t(X_t^*-\mu_t^*)+\log(1-X_t)-\psi(\tau(1-\mu_t))+\psi(\tau)\Big) \\
    & =  \mu_t\mathbbm{E}\Big((X_t^*-\mu_t^*)\Big) + \mathbbm{E}\Big(\log(1-X_t)-\psi(\tau(1-\mu_t))+\psi(\tau)\Big) \\
    &= 0 .
\end{aligned}
\end{equation*}
Similarly, we can calculate the expectations of the rest elements of $G$ and it is easy to see that they are all 0. That is,
\begin{equation*}
    \mathbbm{E}(G(X_t, \bm{\eta}_0|\mathcal{C}_t)) = 0.
\end{equation*}
$G(X_t, \bm{\eta}_0)$ is a MDS. The sum of $G$ based on the $m$ observations and the true parameter $\bm{\eta}_0$, $S_m(X_t, \bm{\eta}_0)$, is a martingale.
\end{proof} 

In \cite{brown1971martingale}, the author proposes certain conditions under which MCLT holds. In order to guarantee the conditions hold, we first show that $G$ has bounded fourth moment. For the sake of simplicity, we denote $G(X_t,\bm{\eta}|\mathcal{C}_t)$ as $G(X_t,\bm{\eta})$ onward.
\begin{lemma}
\label{lemma: bounded fourth moment of G}
 $\mathbbm{E}(|G(X_t,\bm{\eta})|^4)<\infty$. 
\end{lemma}
 The proof can be found in Section \ref{proof: bounded fourth moment}.

\begin{theorem}
\label{theorem: CLT}
The PMLE $\hat{\bm{\eta}}$ is asymptotically normal.
\begin{equation*}
    \sqrt{m}(\hat{\bm{\eta}}-\bm{\eta}_0) \overset{\mathcal{D}}{\rightarrow} N(\bm{0},(-\mathbbm{E}\triangledown G(X_1,\bm{\eta}_0))^{-1})\text{, as } m \rightarrow \infty.
\end{equation*}
\end{theorem}
The proof can be found in Section \ref{proof: CLT}.
The closed-form expression can be found in Section \ref{subsec: Closed-form expressions of the partial log-likelihood function, the score vector and the Hessian matrix}.
\section{Sequential change-point detection for the generalized Beta AR(1) model and the asymptotic results under the null hypothesis}
\label{sec: monitoring scheme and null}
\subsection{Monitoring procedure}
In this section, we will construct close-end sequential monitoring procedure for the generalized Beta AR(1) model based on the score vector defined as (\ref{eq: score vector}). 
\begin{assumption}
\label{assumption: no contamination assumptions}
(a) There is no change point in the exogenous variable $\bm{W}_t$. \\
(b) There is no change point in the first $m$ observations. 
\end{assumption}
Assumption \ref{assumption: no contamination assumptions} (a) indicates that we only discuss the case when any parameters in $\bm{\eta}$ change. Assumption \ref{assumption: no contamination assumptions} (b) (i.e. the non-contamination assumption) guarantees that given the first $m$ observations, we can estimate the parameter vector $\hat{\bm{\eta}}$, which has been proved consistent and asymptotically normally distributed with mean $\bm{\eta}_0$ as $m \rightarrow \infty$. Then the sequential monitoring procedure launches from the $m+1^{\text{th}}$ point. For the close-end monitoring procedure, we stop monitoring when there is no change point detected among the $Nm$ new observations ($N$ is fixed and known). The score vector 
\begin{equation*}
    S_{m,k}(\hat{\bm{\eta}}) = \sum_{t=m+1}^{m+k} G(X_t,\hat{\bm{\eta}})
\end{equation*}
keeps updating as the number of new observations, $k$, increases until it reaches the end point $Nm$. If there is no change point in the monitoring range, $S_{m,k}(\hat{\bm{\eta}})$ should be reasonably close to 0. Therefore, the basic idea for the sequential change-point detection is that we reject the null hypothesis that there is no change point within the monitoring range if the test statistic exceeds a certain threshold, which is controlled by the desired significance level.

Following this idea, suppose that $\{X_t\},t=1,2,...,m,$ follows model (\ref{equation: model definition}) with parameter $\bm{\eta}_0$. For the new observations, we denote $\{X_t\}, t=m+1,m+2,...,m+Nm$, follows model (\ref{equation: model definition}) with parameter $\bm{\eta}_t$, then we can propose test statistics with respect to $S_{m,k}(\hat{\bm{\eta}})$ to test the hypotheses
\\

$H_0: \bm{\eta}_t = \bm{\eta}_0 \text{ for all }t \in \{m+1,m+2,...,m+Nm\},$ 

 v.s.  
 
$H_a: \exists k^*\text{ such that } \bm{\eta}_t = \bm{\eta}_0 \text{ for } m+1\leq t< m+k^* <m+Nm \text{, and }\bm{\eta}_t \neq \bm{\eta}_0 \text{ for } m+k^*\leq t\leq m+Nm.$
\\
Due to the elegant closed-form expression of $S_{m,k}(\hat{\bm{\eta}})$ shown as (\ref{eq: score vector}), we continue to use the monitoring scheme proposed in \cite{kirch2015use} that the null is rejected when 
\begin{equation}
\label{equation: test statistic}
    w^2(m,k)S_{m,k}(\hat{\bm{\eta}})^T\bm{A}S_{m,k}(\hat{\bm{\eta}}) \geq c,
\end{equation}
where $w(m,k)$ is a weight function and $\bm{A}$ is a symmetric positive definite matrix. The matrix $\bm{A}$ can be considered as a rescaled matrix. The test statistic would be analogous to the score test statistic if $\bm{A}$ is chosen as the inverse of the Fisher information matrix of the partial log-likelihood function. The weight function is used to adjust the sensitivity of the monitoring procedure and it satisfies the following assumption. 

\begin{assumption}
\label{assumption: weight function}
The weight function satisfies the form
$$
\omega (m,k) = m^{-\frac{1}{2}}\Tilde{\omega}(m,k),
$$
where
$$
\Tilde{\omega}(m,k) = \begin{cases} \rho(k/m),& \text{ for }k \geq a_m \text{ with } a_m/m \rightarrow0\\
                                    0, &\text{ for } k<a_m.\end{cases}
$$
The function $\rho(.)$ is continuous and satisfies 
$$
\underset{t\rightarrow 0} {\lim} t^{\gamma}\rho(t) <\infty \text{ for some } 0\leq \gamma < \frac{1}{2},
$$
and
$$
\underset{t \rightarrow \infty}{\lim} t\rho(t)<\infty.
$$
The tuning parameter $\gamma$ in the weight function determines the sensitivity of the method. The influence of $\gamma$ on the detection results will be further discussed in Section \ref{sec: simulation} and Section \ref{sec: real life data}.\\
\end{assumption} 
The threshold $c$ is chosen so that the following equation is satisfied under $H_0$. 
\begin{equation*}
    \underset{m\rightarrow \infty}{\lim}P(\underset{1\leq k \leq Nm}{\sup}w^2(m,k)S_{m,k}(\hat{\bm{\eta}})^T\bm{A}S_{m,k}(\hat{\bm{\eta}}) \geq c|H_0)  = \alpha.
\end{equation*}
Assume a change point occurs at $k^*$, say the new parameter $\Tilde{\bm{\eta}} \neq \bm{\eta}_0$, the score vector $S_{m,k}(\hat{\bm{\eta}}), k>k^*,$ will be far from 0 under certain regularity conditions (we will discuss them in Section \ref{sec: power analysis}) as more and more observations after $k^*$ are collected. With a large enough $m$, ideally, the test statistic will exceed the threshold with probability 1. That is, 
\begin{equation*}
    \underset{m\rightarrow \infty}{\lim}P(\underset{1\leq k \leq Nm}{\sup}w^2(m,k)S_{m,k}(\hat{\bm{\eta}})^T\bm{A}S_{m,k}(\hat{\bm{\eta}}) \geq c|H_a)  = 1.
\end{equation*}
We will further discuss the power of the test in Section \ref{sec: power analysis}.
\subsection{Asymptotic results under the null hypothesis}
In the rest of this section, we will derive the asymptotic results of $w^2(m,k)S_{m,k}(\hat{\bm{\eta}})^T\bm{A}S_{m,k}(\hat{\bm{\eta}})$ under the null hypothesis.

\begin{lemma}
\label{proposition: CUSUM approximation}
Under Assumption \ref{assumption: properties of W_t}-\ref{assumption: weight function} and under $H_0$, with the monitoring function $G$ and a weight function $\omega(m,k)$, as $m \rightarrow \infty$,
$$
\underset{k\geq 1}{\sup} \; \omega(m,k)||\sum_{t=m+1}^{m+k} G(X_t,\hat{\bm{\eta}})-(\sum_{t=m+1}^{m+k}G(X_t,\bm{\eta}_0)-\frac{k}{m}\sum_{t=1}^m G(X_t,\bm{\eta}_0))|| = o_p(1).
$$
\end{lemma}

The proof of this lemma can be found in Section \ref{app: proof of CUSUM approximation}. It follows the proof for Proposition 5.2 in \cite{kirch2015use}. This lemma provides us the insight how we can use $S_{m,k}(\hat{\bm{\eta}})$ to estimate $\sum_{t=m+1}^{m+k}G(X_t,\bm{\eta}_0)-\frac{k}{m}\sum_{t=1}^m G(X_t,\bm{\eta}_0)$, which is a type of cumulative sum (CUSUM) statistic commonly used in change point detection, e.g. \cite{lee2003cusum}. 

\begin{proposition}
\label{proposition: wiener processes}
Under Assumption \ref{assumption: properties of W_t}-\ref{assumption: no contamination assumptions} and under $H_0$, for any fixed $N>0$, as $m \rightarrow \infty$, the process
\begin{equation*}
    \frac{1}{\sqrt{m}}\sum_{t=1}^{\floor{ms}}G(X_t,\bm{\eta}_0), \text{ for }0<s \leq N,
\end{equation*}
converges weakly to a $(l+3)$-dimensional Wiener process $W(s)$ with covariance matrix 
\begin{equation*}
    \bm{\Sigma} = \mathbbm{E}\big(G(X_1,\bm{\eta}_0)G^T(X_1,\bm{\eta}_0)\big).
\end{equation*}
$W(s)$ is defined on $D = D[0,N]$, the space of $(l+3)$-dimensional functions that are right-continuous and have left-hand limits, with $s \in [0,N]$\footnote{The details of the definition of the space can be found \cite{oodaira1972functional}.}. 
\end{proposition}
The proof is in Section \ref{app: proof of wiener process}. 

Here we follow the definition of high-dimensional Wiener process on page 49 of  \cite{evans2012introduction}, from which the standardized $l$-dimensional Wiener process, $\bm{\Sigma}^{-1/2}W(s)$, can be considered as a vector of $n$ standardized independent 1-dimensional Wiener process. Analogous to the definition of 1-dimensional Wiener process, at $s$, the $l$-dimensional Wiener process $W(s) \sim N(\bm{0}, s \bm{\Sigma})$. Also, the increment $W(s+u)-W(s), u \geq 0$ is independent with the past values $W(t), t \leq s$. 

\begin{theorem}
\label{theorem: null}
Under Assumption \ref{assumption: properties of W_t}-\ref{assumption: weight function} and under $H_0$, for any fixed symmetric positive definite matrix $\bm{A}$, the distribution of the test statistic 
\begin{equation*}
\begin{aligned}
    \underset{m \rightarrow \infty}{\lim} P\left(\underset{1\leq k \leq Nm}{\sup} w^2(m,k)S_{m,k}(\hat{\bm{\eta}})^T\bm{A}S_{m,k}(\hat{\bm{\eta}}) \leq c    \right) \\
     = &P \left( \underset{0<s\leq N}{\sup} \rho^2(s)(W_1(s)-sW_2(1))^T\bm{A}(W_1(s)-sW_2(1))\leq c \right),
\end{aligned}
\end{equation*}
where $W_1(.)$ and $W_2(.)$ are independent Wiener processes with the same covariance matrix $\bm{\Sigma}$.
\end{theorem}
The proof can be found in Section \ref{app: proof of the null}.

From Lemma \ref{proposition: CUSUM approximation}, we know that $w(m,k)S_{m,k}(\hat{\bm{\eta}})$ can be used to estimate $$w(m,k)(\sum_{t=m+1}^{m+k}G(X_t,\bm{\eta}_0)-\frac{k}{m}\sum_{t=1}^m G(X_t,\bm{\eta}_0)),$$ which is based on the unknown true parameter vector $\bm{\eta}_0 = (\tau_0,\bm{\beta}_0)$. As $m \rightarrow \infty$, from Proposition \ref{proposition: wiener processes}, $\frac{1}{\sqrt{m}}\sum_{t=1}^{\floor{ms}}G(X_t,\bm{\eta}_0)$ converges in distribution to a Wiener process. Therefore, we prove that $w(m,k)S_{m,k}(\hat{\bm{\eta}})$ converges in distribution to a linear combination of the two independent Wiener processes $W_1(.)$ and $W_2(.)$.

\section{Power analysis of the hypothesis test}
\label{sec: power analysis}
In this section, we will prove that when a change point occurs, as $m \rightarrow \infty$, the probability of rejecting the null hypothesis is 1. 
We first propose some assumptions for the change point and the process after the change point.
\begin{assumption}
\label{assumption: power analysis}
(a) The change point happens at $m+k^*$, with $k^* = \lfloor m \nu \rfloor$, where $l<\nu<N$. $l$ and $N$ are fixed and free from $m$. \\
(b) There exists a $x_0$, $\nu<x_0<N$, and a compact neighborhood of $x_0$, $U_{x_0}$, such that $\frac{\lfloor mx_0 \rfloor}{m} \in U_{x_0}$ and $\underset{x \in U_{x_0}}{\inf}\rho(x)>0$, where $\rho(.)$ is the second term of the weight function $
w(m,k)$. \\
(c) Denote the process after the change point as $\{X_t^*\}, t = m+k^*+1,...,m+Nm$. There exists a $(l+3)$-dimensional constant vector $\bm{E}_{G,H_a}$ such that
\begin{equation*}
    \frac{1}{(x_0-\nu)m}\sum_{t=m+k^*+1}^{m+\lfloor mx_0 \rfloor}G(X^*_t, \hat{\bm{\eta}}) \overset{P}{\rightarrow}\bm{E}_{G,H_a},
\end{equation*}
and $\bm{A}^{1/2}\bm{E}_{G,H_a} \neq \bm{0}$.
\end{assumption}
(a) guarantees that the change point would not happen right after the monitoring procedure starts. (c) first makes sure that after the change point $k^*$, $G$ is still ergodic. Also, (b) and (c) together guarantee that the test statistic after the change point would not be equal to zero. (c) excludes the cases that $\bm{E}_{G,H_a} = \bm{0}$ after the change point. In other words, the change should be identified by using the score vector $G$ based on the PMLE $\hat{\bm{\eta}}$ estimated based on the $m$ initial non-contaminated observations. In the meantime, it constrains the choice of $\bm{A}$ so that $\bm{A}^{1/2}\bm{E}_{G,H_a} \neq \bm{0}$. However, it can be naturally satisfied when $\bm{A}$ is a positive definite matrix, which is required in Theorem \ref{theorem: null}.
\begin{theorem}
\label{theorem: power analysis}
Under Assumption \ref{assumption: properties of W_t}-\ref{assumption: power analysis}, the power of the hypothesis test based on (\ref{equation: test statistic}) converges to 1 as $m\rightarrow \infty$. That is,
\begin{equation*}
    \underset{m \rightarrow \infty}{\lim}P\left(\underset{1\leq k \leq Nm}{\sup} w^2(m,k)S_{m,k}(\hat{\bm{\eta}})^T\bm{A}S_{m,k}(\hat{\bm{\eta}}) \geq c  |H_a  \right) = 1.
\end{equation*}
\end{theorem}

The proof can be found in Section \ref{app: proof of power analysis}.

We can eventually reject the null hypothesis if there exists a change point satisfying Assumption \ref{assumption: power analysis}.

\section{Simulation study}
\label{sec: simulation}
In Section \ref{sec: Model definition}, we prove the consistency and asymptotic normality of the estimated parameter vector $\hat{\bm{\eta}}$ based on the first $m$ observations, as $m \rightarrow \infty$. Limiting results for the test statistic under $H_0$ and $H_a$ are shown in Section \ref{sec: monitoring scheme and null} and Section \ref{sec: power analysis}, respectively. In this section, we run simulation studies to verify the properties of $\hat{\bm{\eta}}$ and the asymptotic behaviors of the test statistic under $H_0$ and $H_a$.
\subsection{Consistency and asymptotic normality of $\hat{\bm{\eta}}$}
\label{subsec: simulation-consistency and normality}
To assess the consistency and the asymptotic normality of the estimated parameter vector, we first generate compositional time series $\{X_t\}$ from the model
\begin{equation}
\label{equation: simulation consistency model}
X_t \sim \text{Beta} (100 \mu_t, 100(1-\mu_t)).
\end{equation}
The exogenous variable $\{W_t\}$ follows an AR(1) model
\begin{equation*}
    W_t = -0.1 W_{t-1}+\epsilon_{t}, \epsilon_t \sim N(0,1).
\end{equation*}
To guarantee the boundness of $\{W_t\}$ in order to satisfy Assumption \ref{assumption: properties of W_t}, we set up a lower bound (-10) and a upper bound (10) for $W_t$ and the transformed variable is defined as 
$$W_t^* = \min(\max(-10, W_t), 10).$$
The choices of the lower bound and the upper bound can be flexible. The transformed variable $W^*_t$ can be unchanged if we choose a small enough lower bound and a large large upper bound. 

In this case, we choose $\mathcal{A}(x) = \text{logit}(x^*)$ where $x^* = \min\{\max(c,x),1-c\}$, which is one of the options of $\mathcal{A}(x)$ discussed in Section \ref{sec: Model definition} and has been investigated in the literature, e.g. \cite{woodard2010stationarity}. 
The conditional mean $\mu_t$ is determined by
\begin{equation}
\label{equation: determination of mu_t}
        \text{logit}(\mu_t) = -0.6+0.1 \text{logit}(X_{t-1}^*)+0.1 W_t^*,
\end{equation}
with the transformed variable $X_t^*$ defined as $$X_{t}^* = \min(\max(0.01, X_t), 0.99).$$ 
That is, we generate $\{X_t\}$ from a generalized Beta AR(1) model defined by (\ref{equation: model definition}) and (\ref{eq: link function of GBETAAR}), with parameter vector $\bm{\eta}_0 = (\tau_0, \bm{\beta}_0) = (\tau_0, \varphi_0, \varphi_1, \phi) = (100, -0.6, 0.1, 0.1)$ and threshold $c = 0.01$. 

In order to check the consistency of $\hat{\bm{\eta}}$, we implement three experiments with increasing numbers of observations. For each experiment, we generate 100 samples following (\ref{equation: simulation consistency model}) and (\ref{equation: determination of mu_t}). The three experiments include samples with the number of observations $m = 1000, 2000, 3000$, respectively. We are able to estimate $\hat{\bm{\eta}}$ based on every sample. Under consistency, $\hat{\bm{\eta}}$ should be closer to the true parameter $\bm{\eta}_0$ as the number of observations increases. Therefore, for each experiment, we summarize the mean square error (MSE) of the $i$th element of the PMLEs obtained from the 100 samples, 
\begin{equation*}
    \text{MSE}_i = \frac{1}{100} \sum_{k=1}^{100}(\hat{\bm{\eta}}_{k,i}-\bm{\eta}_{0,i})^2,
\end{equation*}
in order to check the consistency of $\hat{\bm{\eta}}.$ The complete results are listed in Table \ref{table: consistency}. As $m$ increases, the MSEs for all the four parameters decrease, which validates the consistency of $\hat{\bm{\eta}}.$
\begin{table}[!h]
\centering
\begin{tabular}{l|llll}
\hline
MSE        & $\tau_0$  & $\varphi_0$ & $\varphi_1$ & $\phi$  \\ \hline
$m = 1000$ & 542.78966 & 0.20006     & 0.12186     & 0.00436 \\
$m = 2000$ & 12.60172  & 0.00922     & 0.00505     & 0.00014 \\
$m = 3000$ & 7.85109   & 0.00168     & 0.00128     & 0.00003 \\ \hline
\end{tabular}
\caption{MSEs of the PMLEs from the three experiments.}
\label{table: consistency}
\end{table}

Thanks to the ergodicity of the process, we are able to estimate the variance-covariance matrix of $\hat{\bm{\eta}}$ by generating a sample with enough observations. Following (\ref{equation: simulation consistency model}) and (\ref{equation: determination of mu_t}), we generate a sample with 100,000 observations. According to Theorem \ref{theorem: CLT}, as $m \rightarrow \infty$, the variance-covariance matrix of $\sqrt{m}(\hat{\bm{\eta}}- \bm{\eta}_0)$ converges to $(-\mathbbm{E}\triangledown G(X_1, \bm{\eta}_0))^{-1}$. Denote the estimated matrix as $\hat{\bm{\Sigma}}^{-1}_0,$ with $n = 100,000$,
\begin{equation*}
    \hat{\bm{\Sigma}}^{-1}_0 = \Big (\frac{1}{n}\sum_{t=1}^n (-\triangledown G(X_t, \hat{\bm{\eta}})) \Big) ^{-1} \overset{P}{\rightarrow} (-\mathbbm{E}\triangledown G(X_1, \bm{\eta}_0))^{-1}.
\end{equation*}
The expression of $\triangledown G(X_t, \hat{\bm{\eta}})$ can be found in Section \ref{subsec: Closed-form expressions of the partial log-likelihood function, the score vector and the Hessian matrix}. For checking the asymptotic normality, we first generate 1,000 samples from the model (\ref{equation: simulation consistency model}). For each sample, there are 1,001 observations $(X_0,X_1,...,X_{1000})$. That is, $m=1000$. Then, for sample $i$, we can get the PMLE $\bm{\hat{\eta}}_i$. We use the multivariate normality test introduced in \cite{jiang2015likelihood} (Test 2.5) to test the null hypothesis
\begin{equation*}
    \hat{\bm{\eta}}_i \sim N(\bm{\eta}_0, \frac{1}{m}\hat{\bm{\Sigma}}^{-1}_0).
\end{equation*}
The $p$-value of the test based on the samples is 0.71, which verifies the asymptotic normality of $\hat{\bm{\eta}}.$ For each parameter, we also get the quantile-quantile plot (Q-Q plot) for visualizing the marginal distribution of the parameters. Theoretically, the marginal distribution of each parameter is a normal distribution. From the four Q-Q plots shown in Figure \ref{fig: qqplots}, we can see that the marginal distributions can be considered as normal distributions.

\begin{figure}[!h]
    \centering
    \includegraphics[width=0.4\textwidth]{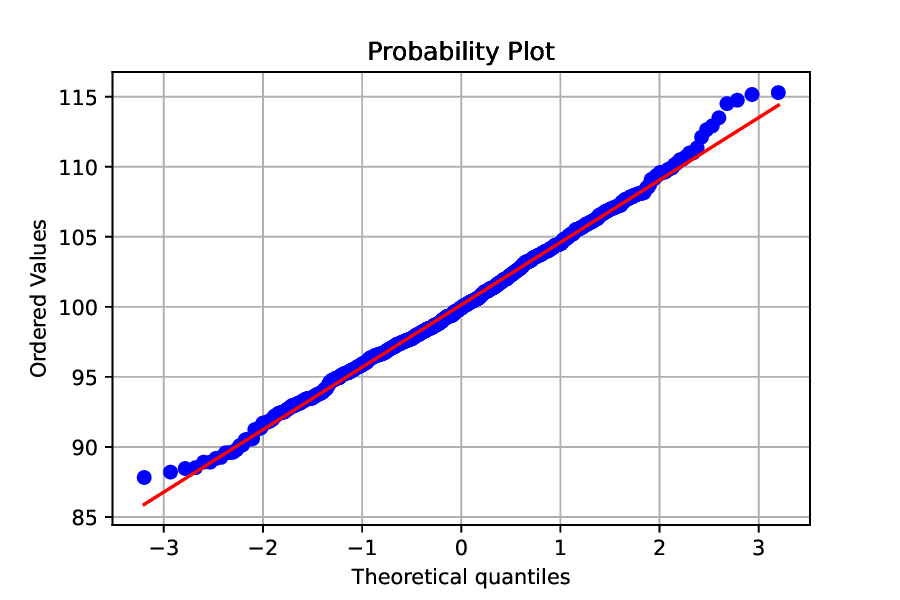}
    \includegraphics[width=0.4\textwidth]{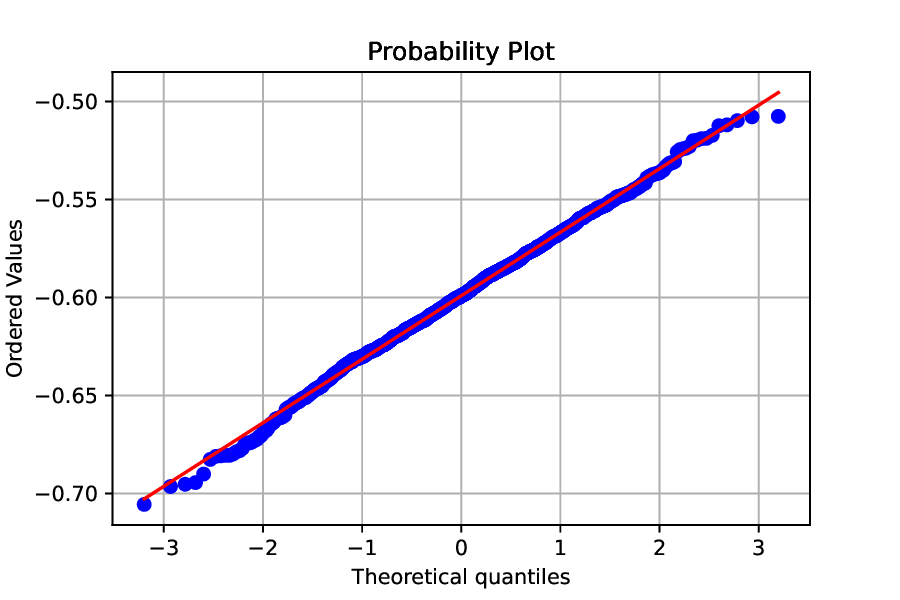}
    \includegraphics[width=0.4\textwidth]{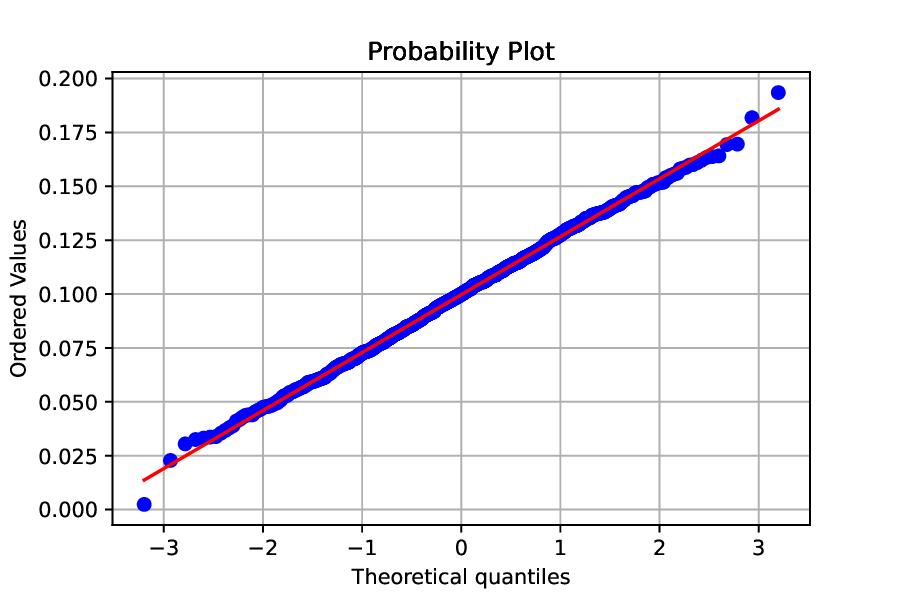} 
    \includegraphics[width=0.4\textwidth]{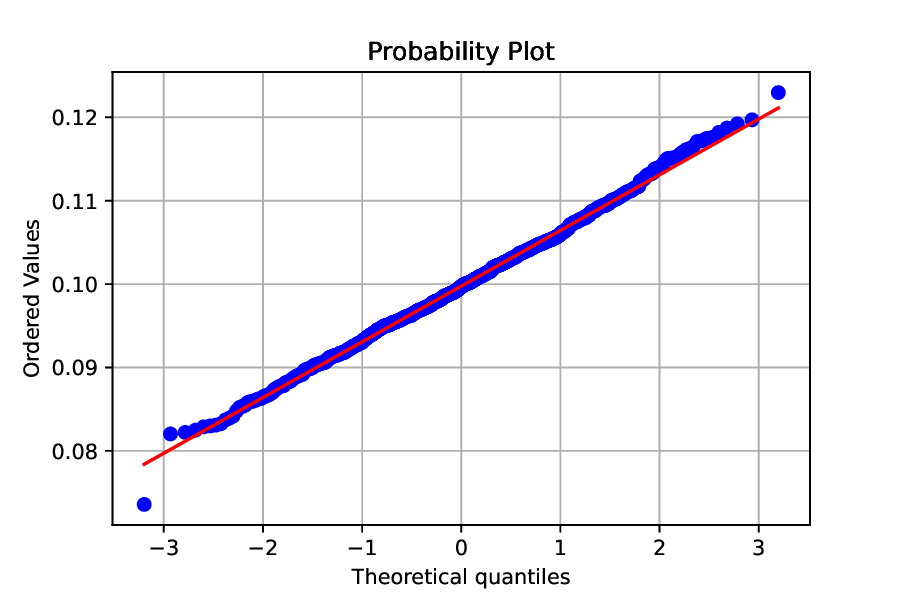} 
    \caption{The Q-Q plots of the four estimated parameters based on the 1000 replications. The true parameter vector is [100, -0.6, 0.1, 0.1].}
\label{fig: qqplots}
\end{figure}

\subsection{Asymptotic results under the null}
\label{sec: asymptotic results under the null}
We use the same model (\ref{equation: simulation consistency model}) to check the validity of the asymptotic distribution under the null hypothesis, as shown in Theorem \ref{theorem: null}. As mentioned in Section \ref{sec: monitoring scheme and null}, the test statistic would be analogous to the score test statistic if $\bm{A} = \bm{\Sigma}^{-1}_0,$ which is the inverse of the information matrix of the partial log-likelihood function. The idea has been used in building test statistics for other change-point detection problems, e.g. \cite{huvskova2005monitoring} and \cite{hudecova2013structural}. In this case, we assign the estimated inverse of the theoretical information matrix $\hat{\bm{\Sigma}}^{-1}_0$ calculated in Section \ref{subsec: simulation-consistency and normality} to the $\bm{A}$ in the test statistic. 

We choose the weight function 
\begin{equation}
    \label{eq: definition of the weight function}
    \omega(m,k, \gamma) = m^{-\frac{1}{2}}(1+\frac{k}{m})^{-1}(\frac{k}{m+k})^{-\gamma}, 0\leq \gamma< \frac{1}{2},
\end{equation}
for the simulations in this chapter. Following the notations in Assumption \ref{assumption: weight function}, we denote 
\begin{equation*}
    \rho(s,\gamma) = s^{-\gamma}(s+1)^{\gamma-1},
\end{equation*}
then 
\begin{equation*}
    \omega(m,k, \gamma) = m^{-\frac{1}{2}}\rho(k/m, \gamma), 0\leq \gamma< \frac{1}{2}.
\end{equation*}
With a fixed $\gamma$ and the significance level $\alpha$, the threshold $c(\gamma,\alpha)$ can be determined so that
$$P(\underset{0< s \leq N}{\sup}\rho^2(s,\gamma)(W_1(s)-sW_2(1))^T\bm{A}(W_1(s)-sW_2(1))\leq c(\gamma,\alpha)) = 1-\alpha.$$
According to Theorem \ref{theorem: null}, $W_1(.)$ and $W_2(.)$ are Wiener processes with covariance matrix equal to $\bm{\Sigma}_0$.

In order to investigate the influence of $\gamma$ on the sensitivity of the change point detection approach, we choose $\gamma = 0, 0.25, 0.4$ for the simulation. With $N = 3$, for every $\gamma$, 10,000 samples of $\rho^2(s,\gamma)(W_1(s)-sW_2(1))^T\bm{A}(W_1(s)-sW_2(1)), 0<s\leq N$ are generated, and $\hat{\bm{\Sigma}}_0$ is used as the estimated covariance matrix of $W_1(.)$ and $W_2(.)$. Following the generating method we mentioned in Theorem \ref{theorem: null}, we choose $m = 1000$ and generate $N*m = 3000$ i.i.d. normal random vectors following $N(\bm{0}, \hat{\bm{\Sigma}}_0)$. 
The cumulative sum of the first $k$ random variables divided by $\sqrt{m}$ simulates $W_1(k/m), k = 1,2,...,Nm.$  Similarly, we are able to simulate a $W_2(1)$ by generating a random vector following $N(\bm{0}, \hat{\bm{\Sigma}}_0)$. For every iteration, we repeat the procedure to build $\rho^2(s, \gamma)(W_1(s)-sW_2(1))^T\bm{A}(W_1(s)-sW_2(1)), s = \frac{1}{m}, \frac{2}{m},...,\frac{Nm}{m}$, and record the supremum over $s$. Then, the $(1-\alpha)$ quantile of the supremum of the 10,000 samples is considered as the threshold $c(\gamma, \alpha)$.
In Table \ref{table:null c }, we list the threshold $c(\gamma,\alpha)$ for different $(\gamma,\alpha)$ with $N = 3$.
\begin{table}[!htb]
\centering
\begin{tabular}{lllll}
\hline
$\alpha$        & 0.1    & 0.05   & 0.025  & 0.01    \\ \hline
$\gamma = 0$    & 6.7396 & 7.9931 & 9.1888 & 10.5312  \\
$\gamma = 0.25$ & 8.4479 & 9.9127 & 11.3129 & 13.0243  \\
$\gamma = 0.4$  & 10.4888 & 12.0926 & 13.6117 & 16.0009 \\ \hline
\end{tabular}
\caption{$c(\gamma,\alpha)$ under different $\gamma$ and significance level $\alpha$}\label{table:null c }
\end{table}

\noindent
With $N = 3$, $m=500, 1000,1500$ and $\gamma = 0,0.25,0.4$, the empirical probability that the test statistic crossing the threshold 
$$\underset{1\leq k \leq Nm}{\sup} w^2(m,k)S_{m,k}(\hat{\bm{\eta}})^T\bm{A}S_{m,k}(\hat{\bm{\eta}}) \geq c(\gamma, \alpha)$$
are obtained based on 5,000 repetitions, and the comparisons of the empirical probabilities and the theoretical $\alpha$'s are listed in Table \ref{table:null alpha compare}.
\begin{table}[!htb]
\centering
\begin{tabular}{llllll}
\hline
\multicolumn{2}{c}{$\alpha$}                                & 0.1    & 0.05                       & 0.025  & 0.01   \\ \hline
\multicolumn{1}{c}{\multirow{3}{*}{$\gamma = 0$}} & $m=500$ & 0.1158 & 0.0644& 0.0398& 0.0184 \\
\multicolumn{1}{c}{}                              & $m=1000$ & 0.1018 & 0.0574 & 0.0328 & 0.0162 \\
\multicolumn{1}{c}{}                              & $m=1500$ & 0.1062 & 0.0576 & 0.0298 & 0.0152 \\ \hline
\multirow{3}{*}{$\gamma = 0.25$}                  & $m=500$ & 0.1252& 0.0696& 0.0398& 0.0194 \\
                                                  & $m=1000$ & 0.1106& 0.0592& 0.0358& 0.0170 \\
                                                  & $m=1500$ & 0.1086& 0.0586& 0.0296& 0.0138 \\ \hline
\multirow{3}{*}{$\gamma = 0.4$}                   & $m=500$ & 0.1756&0.1168& 0.0732& 0.0366 \\
                                                  & $m=1000$ & 0.1480 & 0.0954& 0.0594& 0.0266 \\
                                                  & $m=1500$ & 0.1456& 0.0816& 0.0488& 0.0218\\ \hline
\end{tabular}
\caption{The empirical probabilities with different $m$ and $\gamma$}\label{table:null alpha compare}
\end{table}

According to Table \ref{table:null alpha compare}, we can see the influence of $\gamma$ and the size of the initial observations $m$ on the results of the change point detection method.
\begin{itemize}
    \item With a fixed $\gamma$, as the amount of initial observations $m$ increases, the empirical probability approaches the nominal significance level $\alpha$. 
    
    It coincides with the theoretical results we proposed in Theorem \ref{theorem: null}. As $m \rightarrow \infty$, the empirical probability should converge to the theoretical probability.
    
    \item As $\gamma$ increases, it requires larger $m$ to guarantee the empirical probability is close to the nominal significance level $\alpha$. 
    
    For example, with the nominal significance level $\alpha = 0.1$, the empirical probability when $\gamma = 0$ and $m=1000$ has been already around 0.1. However, when $\gamma = 0.4$ and $m=1000$, the empirical probability is 0.1480, still far away from 0.1. In other words, with the same $m$, test statistic with larger $\gamma$ tends to provide more sensitive detection results. We will discuss this in Section \ref{sec: asymptotic results under the alternative}.
\end{itemize}

\subsection{Asymptotic results under the alternative}
\label{sec: asymptotic results under the alternative}
In order to examine the power of the test, we also generate processes including change points to see whether the test can detect the change point and how sensitive the test is. There are two simulation strategies. 
First, with different amounts of initial observations $m=100,500,1000$ and $N=3$, we generate 5,000 ($K$) processes. For every process, the first $m+50$ data points follow the Beta AR(1) model (\ref{equation: simulation consistency model}) and (\ref{equation: determination of mu_t}), and the following $Nm-50$ data points follow a Beta AR(1) model with changed parameter vector $\bm{\eta}' = (\tau', \bm{\beta}') = (\tau', \varphi'_0, \varphi'_1, \phi') = (100, -0.6, 0.2, 0.1)$. Other parameters are the same as the original model except the parameter vector $\bm{\eta}'$. That is, the change point occurs at $k^* = 50$. The parameter vector is estimated based on the first $m$ observations, and then the monitoring procedure launches from $m+1$. With $\gamma = 0, 0.25, 0.4$, we record the first time point $k$ ($1\leq k \leq Nm$) that $w^2(m,k,\gamma)S_{m,k}(\hat{\bm{\eta}})^T\bm{A}S_{m,k}(\hat{\bm{\eta}}) \geq c(\gamma,0.05)$, which is the time point we reject the no change point null hypothesis with significance level 0.05, if the $k$ exists. 

Also, we run simulation studies with $m=100$ and change point $k^* = 10,30,50$ to compare the influence of the distance between the launch point and the change point on detection sensitivity. 

For each experiment, we evaluate the performance of the sequential change point detection method based on the three metrics.
\begin{itemize}
    \item (M1) The distance between the mean of detected change point $k$ and the ground truth $k^*$.
    \item (M2) The percentage of processes that the null hypothesis is rejected.
    \item (M3) The sensitivity of the detection method, defined as
    \begin{equation}
    \label{eq: definition of sensitivity}
        \text{Sensitivity} = \frac{\sum_{i=1}^K \text{Boolean}((k-k^*)>0)}{K} \times 100\%.
    \end{equation}
    Given that the change point occurs at $k^*$, we introduce the metric to measure the percentage of processes with detected $k$ greater than $k^*$. Although we would like to detect the change point as quick as possible, it is also important to control the sensitivity of the monitoring procedure so that
    $H_0$ would not be rejected before $k^*$.
\end{itemize}
Table \ref{table: asymptotic results under the alternative1} shows the three metrics for each $(m,\gamma)$ combination with ground truth change point $k^* = 50$. 

\begin{table}[!htb]
\centering
\begin{tabular}{l|l|lll}
\hline
Metric                & $\gamma$ & $m=100$  & $m=500 $          & $m=1000$ \\ \hline
\multirow{3}{*}{(M1)} & $0$      & 65.46 & 142.99 & 180.83 \\
                      & $0.25$   & 39.68  & 90.67          & 108.86 \\
                      & $0.4$    & 13.66  & 41.33           & 50.36 \\ \hline
\multirow{3}{*}{(M2) (\%)} & $0$      & 86.50  & 100.00          & 100.00 \\
                      & $0.25$   & 88.66  & 100.00          & 100.00 \\
                      & $0.4$    & 90.26  & 100.00          & 100.00 \\ \hline
\multirow{3}{*}{(M3) (\%)} & $0$      & 73.06  & 99.92           & 100.00 \\
                      & $0.25$   & 57.24  & 93.50           & 97.98  \\
                      & $0.4$    & 39.76  & 67.38           & 75.40  \\ \hline
\end{tabular}
\caption{The three metrics for each $(m,\gamma)$ combination with $k^* = 50$.}
\label{table: asymptotic results under the alternative1}
\end{table}
Table \ref{table: asymptotic results under the alternative2} shows the three metrics for each $(k^*,\gamma)$ combination with the same amount of initial observation $m = 100$. 
\begin{table}[!htb]
\centering
\begin{tabular}{l|l|lll}
\hline
Metric                & $\gamma$ & $k^*=10$  & $k^*=30$           & $k^* = 50$ \\ \hline
\multirow{3}{*}{(M1)} & $0$      & 74.33 & 71.43 & 65.46 \\
                      & $0.25$   & 51.28 & 47.26          & 39.68 \\
                      & $0.4$    & 32.66  & 25.71          & 13.66 \\ \hline
\multirow{3}{*}{(M2) (\%)} & $0$      & 92.36 & 89.66         & 86.50 \\
                      & $0.25$   & 93.48  & 91.26         & 88.66 \\
                      & $0.4$    & 93.92  & 92.68         & 90.26 \\ \hline
\multirow{3}{*}{(M3) (\%)} & $0$      & 91.82  & 83.58         & 73.06 \\
                      & $0.25$   & 83.22  & 68.56           & 57.24  \\
                      & $0.4$    & 64.02  & 49.06           & 39.76 \\ \hline
\end{tabular}
\caption{The three metrics for each $(k^*,\gamma)$ combination with $m = 100$.}
\label{table: asymptotic results under the alternative2}
\end{table}

From Table \ref{table: asymptotic results under the alternative1} and Table \ref{table: asymptotic results under the alternative2} we can conclude the influence of different choices of $\gamma$ and $m$ on the final detection results and provide insights on real-life practice.
\begin{itemize}
    \item The detection response speed increases as $\gamma$ increases. 
    
    With fixed $m$, from Table \ref{table: asymptotic results under the alternative1}, the distance between the mean of detected change point $k$ and the ground truth $k^*$ decreases as $\gamma$ increases. That is, the detection method can detect change point sooner with higher $\gamma$. Therefore, in practice, if we assume a change point will occur soon, a relatively large $\gamma$ might be chosen so as to detect the change point as soon as possible. However, in terms of (M3), a synchronous drawback of high $\gamma$ is low sensitivity (defined as (\ref{eq: definition of sensitivity})). Since a test statistic with higher $\gamma$ is more sensitive to possible change points, it is more likely to cause false positive alarms.
    
    \item According to (M2) in Table \ref{table: asymptotic results under the alternative1}, the power of the test goes to 1 as $m$ increases.
    
    The simulation outputs validate the theoretical results proposed in Theorem \ref{theorem: power analysis}. $H_0$ is rejected for all the sample processes when $m = 500$ or $m = 1000$. However, when $m = 100$, there are processes of which the change points are not successfully detected. In practice, when a false negative alarm is really expensive, we suggest to choose $m$ as large as possible.
\end{itemize}

\section{Applications of the Generalized Beta AR(1) Model and the Sequential Change-point Detection Method} 
\label{sec: real life data}
We will apply the generalized Beta AR(1) model to two real-life data sets: 
\begin{enumerate}
    \item (D1) percentage change of the consumer price index for alcoholic beverages in U.S. (referred as CPI-AB onward) from its value 12 months prior, from January 2010 to February 2022,
    \item (D2) daily Covid-19 positivity rate in Arizona from 03/01/2021 to 10/28/2021,
\end{enumerate} and investigate the usage of the model from the two aspects, forecasting ability on (D1) and sequential change-point detection on (D2). 

The flexibility of the choice of $\mathcal{A}(x)$ allows us to fit proportional time series using different $\mathcal{A}(x)$ and choose the most appropriate function for the application. In order to illustrate the influences of different $\mathcal{A}(x)$, for both applications, we fit generalized Beta AR(1) models with three $\mathcal{A}(x)$, 
\begin{itemize}
    \item identity transformation, $\mathcal{A}(x) = x$,
    \item logit transformation, $\mathcal{A}(x) = \log \frac{x^*}{1-x^*}$,
     \item complementary log-log transformation, $\mathcal{A}(x) = \log(-\log(1-x^*))$.
\end{itemize}
$x^*$ is the truncated output defined as $x^* = \min\{\max(c,x),1-c\}$, where $c$ is the customized constant. From the practical point of view, $c$ is used to prevent overflow when $x$ is very close to (or equal to) 0 or 1. Therefore, $c$ can be 0 when none of the observations is close to 0 and 1, which is the case for both (D1) and (D2).

First, in order to indicate the effects of different choices of $\mathcal{A}(x)$ on the forecasting results, we fit generalized Beta AR(1) models to the monthly percentage change of CPI-AB with three choices of $\mathcal{A}(.)$, logit transformation, complementary log-log transformation and identity transformation. We will evaluate the performances of the models with different $\mathcal{A}(x)$ through both in-sample error metrics and out-of-sample error metrics, e.g. mean absolute error (MAE), mean absolute percentage error (MAPE) and so on. The definitions of the metrics will be introduced in Section \ref{sec: forecasting results}.

We will also fit a generalized Beta AR(1) model to the Arizona daily Covid-19 positivity rate time series, and then employ the proposed sequential change-point detection method to explore the possibility of a change point, which may be caused by the Delta variant of Covid-19. During the modeling process and the change point detection process for (D2), we use daily extreme weather indication as the exogenous variable. It is a binary variable indicating the existence of extreme weathers when it is 1, and 0 otherwise. It is based on the fact that extreme weathers affect the Covid-19 test procedure so as to have an influence on the following positivity rate (details will be discussed in Section \ref{sec: Change-point detection application to Arizona daily Covid-19 positivity rate}). We select the best of the three $\mathcal{A}(x)$ by minimizing AIC. Then, we conduct sequential change-point detection based on the estimated model. By comparing the change points detected by the method and the reports about when Delta variant started to dominate in Arizona, we can see how the sequential change-point detection method can be used in practice and provides alerts in areas like epidemiology.

\subsection{Forecasting CPI-AB percentage change using generalized Beta AR(1) model with different $\mathcal{A}(x)$}
\label{sec: forecasting results}
The consumer price index for alcoholic beverages in U.S. (CPI-AB) is an indicator measuring the overall price level of alcoholic beverages. Overall speaking, CPI-AB keeps increasing month by month. The reasons can be inflation, increase taxes on alcoholic beverages, increasing minimum pricing restrictions on alcohol and so on. Meanwhile, CPI-AB displays seasonality. Therefore, we collect the monthly percentage change of CPI-AB from 12 months prior. Because of the overall increasing trend, the percentage is always greater than 0 (and also less than 1). Also, it can remove the seasonality of the original time series. Figure \ref{fig: CPI_AB} shows the percentage change of the CPI-AB \footnote{Data source: FRED \url{https://fred.stlouisfed.org/series/CUSR0000SAF116}.} from 12 months prior from January 2010 to February 2022. 
\begin{figure}[!h]
    \centering
    \includegraphics[width = 0.8\textwidth]{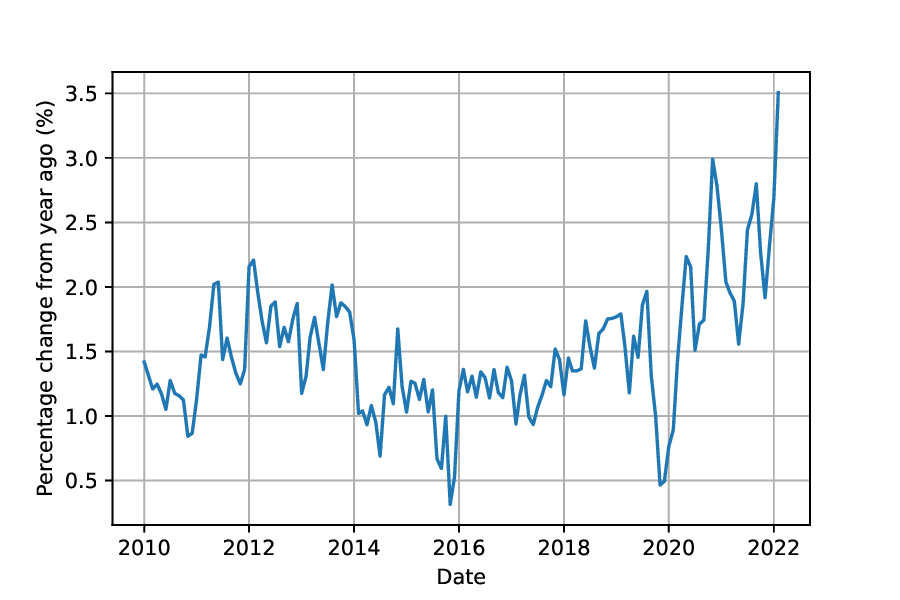}
    \caption{Percentage change of monthly CPI-AB from January 2010 to February 2022.}
     \label{fig: CPI_AB}
\end{figure}

We consider the outputs from January 2010 to December 2014 as training samples (in-sample period) and fit the Beta AR(1) models with the three $\mathcal{A}(x)$. We conduct one-step ahead predictions for the CPI-AB changes since 2015 (out-of-sample period) using the three models. The in-sample performance and the out-of-sample forecasting results are evaluated based on the four metrics.
\begin{itemize}
    \item (L1) Mean absolute error (MAE).
    \begin{equation*}
        \text{MAE} = \frac{\sum_{t=1}^N |X_t-\hat{X}_t|}{N}
    \end{equation*}
    \item (L2) Mean absolute percentage error (MAPE).
    \begin{equation*}
        \text{MAPE} = \sum_{t=1}^N \Big|\frac{X_t-\hat{X}_t}{X_t}\Big| \times 100\%
    \end{equation*}
    \item (L3) Root mean square error (RMSE).
    \begin{equation*}
        \text{RMSE} = \sqrt{\frac{\sum_{t=1}^N (X_t-\hat{X}_t)^2}{N}}
    \end{equation*}
    \item (L4) Coverage percentage (CP).
    \begin{equation*}
            \text{CP} = \frac{\sum_{t=1}^N \text{Boolean}(L_t \leq X_t \leq U_t)}{N} \times 100\%
    \end{equation*}
\end{itemize}
The first three metrics are commonly used in model assessment. In order to compare the forecasting abilities and the robustness of the models, we introduce the metric (L4). Given the customized significance level $\alpha$, for each time $t$, we are able to estimate the $(1-\alpha) \times 100\%$ confidence interval $[L_t, U_t]$ for the one-step ahead prediction for $X_t$. Then, we calculate the percentage of the real outputs covered by the estimated confidence interval. The metrics of the in-sample results and the out-of-sample results are listed in Table \ref{table: CPI-AB in-sample metrics} and Table \ref{table: CPI-AB out-of-sample metrics}, respectively.

\begin{table}[!h]
\centering
\begin{tabular}{l|lllll}
\hline
$\mathcal{A}(x)$ &  (L1)   & (L2) (\%) & (L3)   & (L4) (\%) \\ \hline
$x$                               & 0.0032 & 24.37     & 0.0038 & 91.38     \\
$\log \frac{x^*}{1-x^*}$          & 0.0020 & 14.50     & 0.0025 & 89.66      \\
$\log(-\log(1-x^*))$              & 0.0020 & 14.35     & 0.0025 & 89.66     \\ \hline
\end{tabular}
\caption{In-sample error metrics. \\The significance level is $\alpha = 0.1$, so (L4) is the CP of the 90\% confidence interval. The threshold $c=0$.}
\label{table: CPI-AB in-sample metrics}
\end{table}
\begin{figure}[!h]
    \centering
    \includegraphics[width=0.4\textwidth]{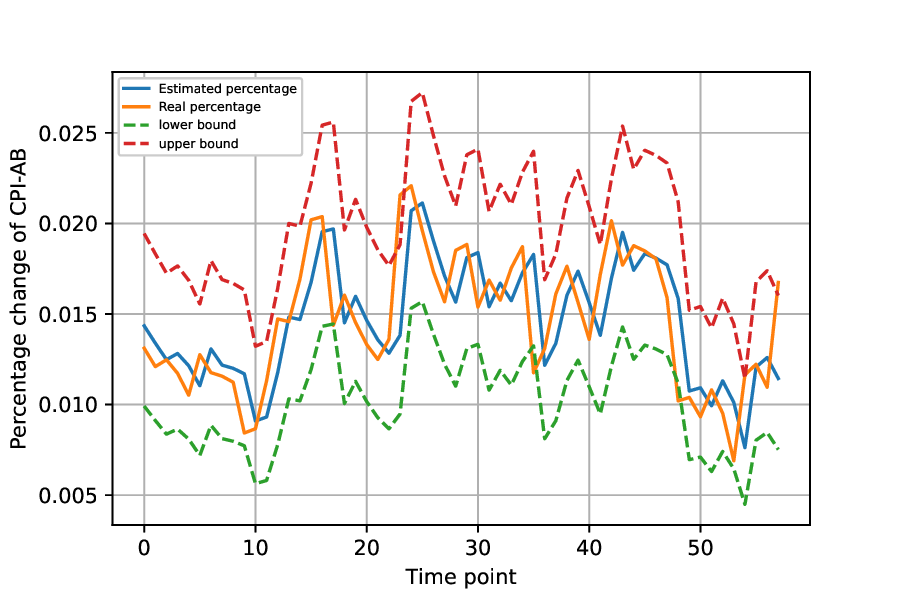} 
    \includegraphics[width=0.4\textwidth]{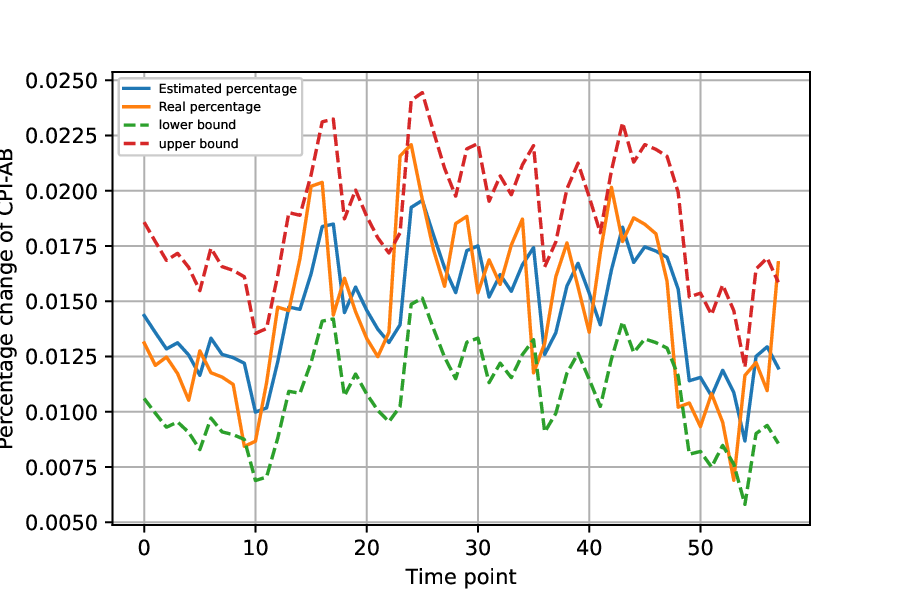}
    \includegraphics[width=0.4\textwidth]{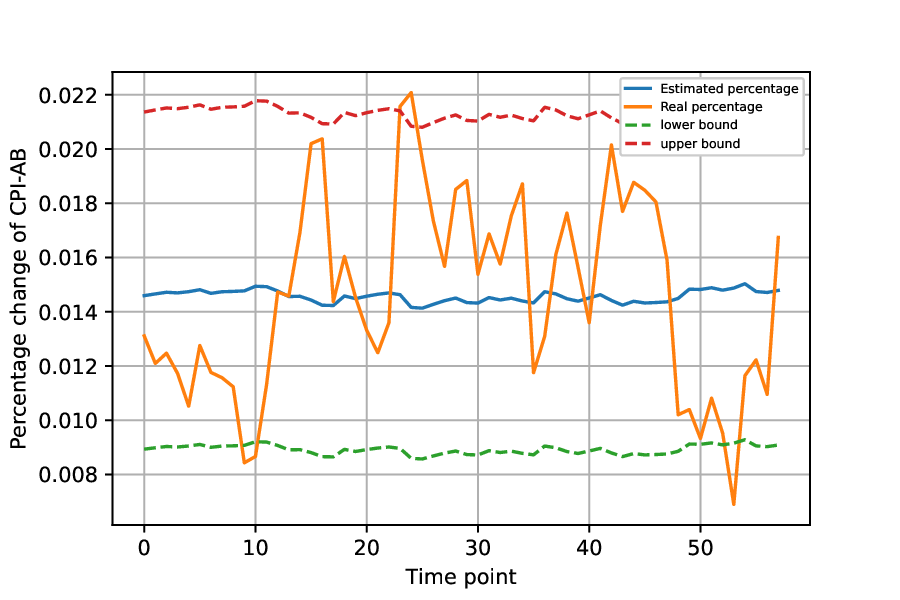} 
    \caption{The in-sample model fitting results and the 90\% confidence intervals. Top left: complementary log-log link; top right: logit link; bottom: identity link. In each plot, the orange solid line is the real data; the blue solid line is the estimated sequence; the red dotted line is the upper bound of the prediction confidence interval; the green dotted line is the lower bound of the prediction confidence interval.}
\label{fig: ci comparison in-sample}
\end{figure}
\begin{table}[!h]
\centering
\begin{tabular}{l|lllll}
\hline
$\mathcal{A}(x)$ &  (L1)   & (L2) (\%) & (L3)   & (L4) (\%) \\ \hline
$x$                               & 0.0046 & 38.51     & 0.0062 & 75.58     \\
$\log \frac{x^*}{1-x^*}$          & 0.0025 & 20.26     & 0.0034 & 81.40      \\
$\log(-\log(1-x^*))$              & 0.0024 & 19.93     & 0.0031 & 87.21     \\ \hline
\end{tabular}
\caption{Out-of-sample error metrics. \\The significance level is $\alpha = 0.1$, so (L4) is the CP of the 90\% confidence interval. The threshold $c=0$.}
\label{table: CPI-AB out-of-sample metrics}
\end{table}

Logit link function is the most commonly used $x$-link function in the literature, e.g. \cite{rocha2009beta}, \cite{scher2020goodness}, \cite{woodard2010stationarity}. However, when the increase and decrease of the previous observation don't have symmetric influences on the current observation, certain asymmetric link functions can be more appropriate options. As a commonly used asymmetric link function, complementary log-log link function is included in this application. The identity link function was used in models such as \cite{wang2011autopersistence}. Therefore, we also consider that as an option. According to the error metrics (L1)-(L3) in Table \ref{table: CPI-AB in-sample metrics}, we can see that the model based on the identity link is not appropriate, and the models based on the logit link and the complementary log-log link have similar in-sample fitting results. We can get the same conclusion from Figure \ref{fig: ci comparison in-sample}. The fitted model based on the identity link is flat and is not informative. But the models based on the other two link functions can reflect the fluctuation of the time series. 

The lack of predictive ability of the model with identity link is also reflected in the out-of-sample results. As shown in Table \ref{table: CPI-AB out-of-sample metrics}, it has high error metrics (L1)-(L3) and has low coverage rate (L4). Figure \ref{fig: ci comparison out-of-sample} shows the prediction results from the models with the complementary log-log link and the logit link. From the plots we can see possible deviations from stationarity roughly after the $80^{\text{th}}$ point. In order to show the differences between the two models more obviously, we include all the points even the sequence displaying a possible deviation from stationarity. Comparing with the model based on the logit link, the model with complementary log-log link outperforms in terms of all the metrics, especially (L4). From Figure \ref{fig: ci comparison out-of-sample} we can see that the model with complementary log-log link provides more robust point estimations and is more adaptive to large oscillations. It leads to an average coverage percentage (87.21\%) closer to the nominal one (90\%), comparing with the counterpart of the model based on the logit link. Therefore, we consider the model with complementary log-log link as a proper model for this application.

\begin{figure}[!h]
    \centering
    \includegraphics[width=0.8\textwidth]{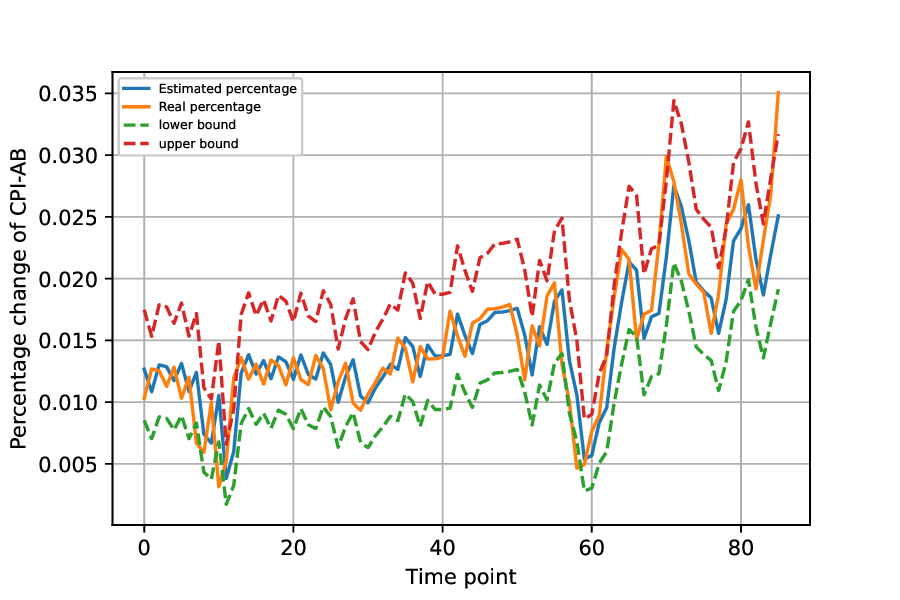} 
    \includegraphics[width=0.8\textwidth]{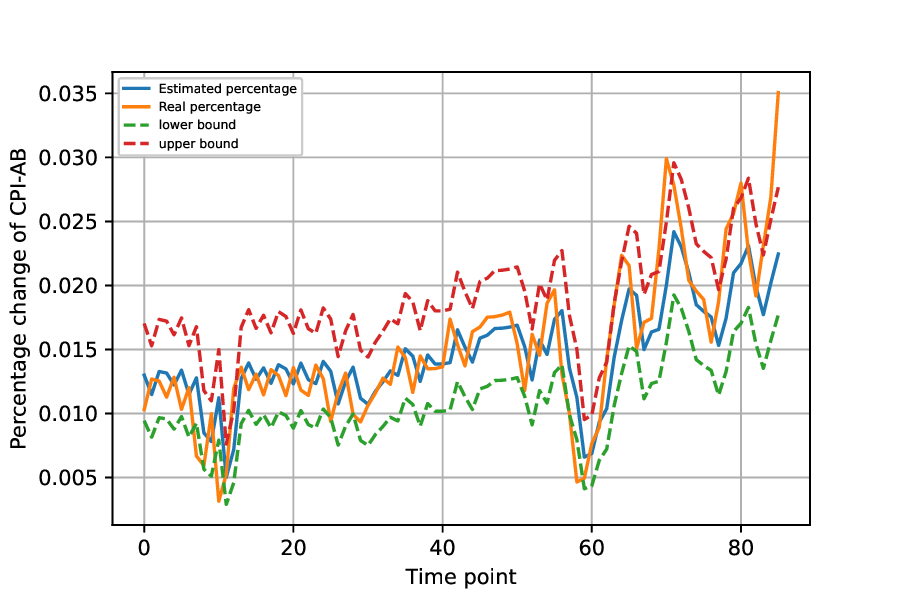}
    \caption{The out-of-sample model fitting results and the 90\% confidence intervals. Top: complementary log-log link; bottom: logit link. In each plot, the orange solid line is the real data; the blue solid line is the estimated sequence; the red dotted line is the upper bound of the prediction confidence interval; the green dotted line is the lower bound of the prediction confidence interval.}
\label{fig: ci comparison out-of-sample}
\end{figure}

\subsection{Change-point detection application to Arizona daily Covid-19 positivity rate}
\label{sec: Change-point detection application to Arizona daily Covid-19 positivity rate}
Based on the nationwide data collected by CDC \footnote{\url{https://covid.cdc.gov/covid-data-tracker}}, the Covid-19 positivity rate was roughly stationary from, roughly speaking, the beginning of March, 2021 to the end of June, 2021. Then, because of the spread of Delta variant, Covid-19 positivity rates started to soar again. In the meantime, we notice that the daily positivity rate is affected by exogenous factors. There are several assumptions indicating that extreme weathers affect the Covid testing procedure, and therefore, influence the daily positivity rate. 
\begin{enumerate}
    \item Testing sites would stop operations due to extreme weathers.
    
    Since some testing sites would be closed because of extreme weathers, it is reasonable to assume that extreme weathers would have an influence on the Covid-19 positivity rate. 
    
    \item Extreme weathers would stop people going outside and got tested. 
\end{enumerate}
In order to detect the time point when the influence of Delta variant started to be reflected on the Covid-19 positivity rate, we would like to fit generalized Beta AR(1) models on the positivity rate time series and see if the sequential change-point detection method can successfully detect a change point, which can shed insights on the time point when Delta variant affects the positivity rate.

The time points of Delta variant outbreaks are different across all the states. Therefore, we only consider the daily Covid-19 positivity rate time series in Arizona. The daily positivity rate time series in Arizona from March 2021 to July 2021 is more stationary compared with other states, which is helpful for model parameter estimation. 

\begin{figure}[!h]
    \centering
    \includegraphics[width = 0.8\textwidth]{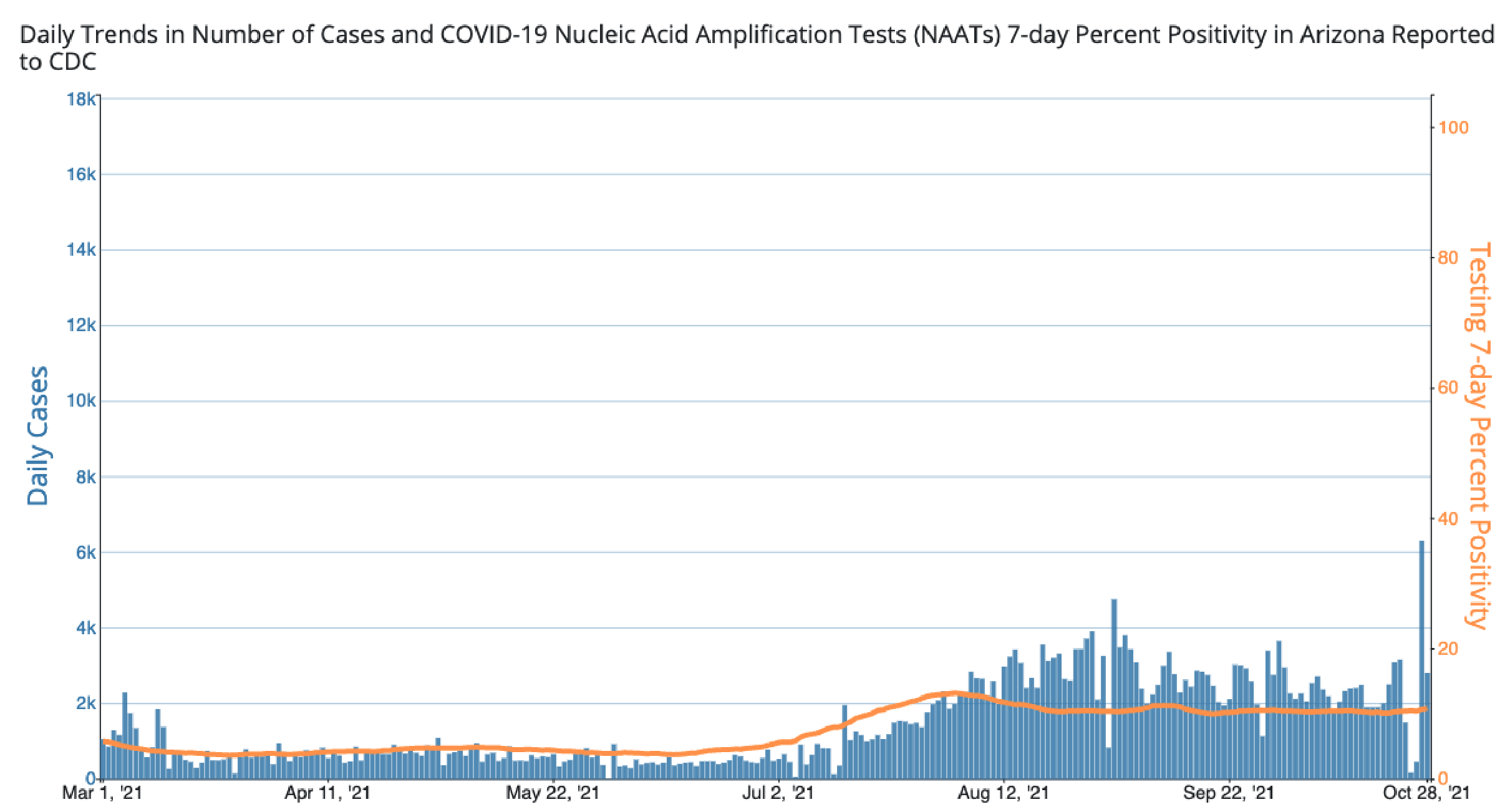}
    \caption{The daily Covid-19 positivity rate in Arizona from 03/01/2021 to 10/28/2021.}
     \label{fig:ts_positiverate_AZ}
\end{figure}

We collect the daily Covid-19 positivity rate time series in Arizona from 03/01/2021 to 10/28/2021. That is, there are overall 242 observations, i.e. the $X_t$ in (\ref{equation: model definition}). Based on the aforementioned assumptions, we would like to investigate the influence of extreme weathers on the daily positivity rate. For the same time period and the same frequency, we collect extreme weather indication time series in Phoenix from National Centers for Environmental Information \footnote{\url{https://www.ncdc.noaa.gov/cdo-web/search}}. The output is 1 if any of the extreme weathers, e.g. fog, thunder, sleet and so on, happens, and otherwise, 0. It is considered as the exogenous variable $W_t$ in (\ref{equation: model definition}).

We use the first $m=100$ (03/01/2021 to 06/09/2021) observations for model selection and parameter estimation, under the assumption that there is no change point existing in the first 100 observations. We fit generalized Beta AR(1) models incorporating the exogenous variable sequence based on three $\mathcal{A}(x)$, identity transformation, logit transformation and complementary log-log transformation. The estimated parameters and the AIC scores for the three models are listed in Table \ref{Table: model comparison}. According to Table \ref{Table: model comparison}, the model with logit $x$-link function has the lowest AIC. So we use the fitted generalized Beta AR(1) model 
\begin{equation}
\begin{aligned}
    X_t &\sim \text{Beta} (2803.541 \mu_t, 2803.541 (1-\mu_t)) \\
\text{logit}(\mu_t) &= -2.215+0.290 \text{logit}(X_{t-1}^*)+0.004W_t
\end{aligned}
\end{equation}
for the following change-point detection. In order to check the influence of $W_t$ on the prediction of $X_t$, we also fit a model excluding $W_t$, with $\mathcal{A}(x)$ as the logit transformation. So the link function follows the format
\begin{equation*}
    \log(\mu_t) = \varphi_0+\varphi_1 \text{logit}(X_{t-1}^*).
\end{equation*}
The corresponding AIC = -574.941. Comparing it with the AIC of the model including exogenous variable, we can conclude that the exogenous variable sequence can indeed improve the model fitting results. 
\begin{table}[!h]
\centering
\begin{tabular}{l|lllll}
\hline
$\mathcal{A}(x)$         & $\tau$ & $\varphi_0$ & $\varphi_1$ & $\phi$ & AIC \\ \hline
$x$                      &  302.834 & -4.036 & 22.411& -36.708&-320.163\\
$\log \frac{x^*}{1-x^*}$ &    2803.541 &-2.215 & 0.290 & 0.004 &-863.571\\
$\log(-\log(1-x^*))$     &   210.512 & -10.747& -2.418&-0.018 &-561.035 \\ \hline
\end{tabular}
\caption{Parameter estimation based on the different $\mathcal{A}(x)$ and the corresponding AIC scores $x^* = \min\{\max(0.01,x),0.99\}$.}
\label{Table: model comparison}
\end{table}

Since the total number of observations is 242 and the amount of initial observations is $m=100$, we consider the rest 142 observations as the new observations after parameter selection, and it is a close-end sequential change-point detection problem with $N = \frac{142}{100} = 1.42.$ We use the same weight function as the one in Section \ref{sec: asymptotic results under the null} with $\gamma = 0, 0.25, 0.4$. With the significance level $\alpha = 0.05$ and $N = 1.42$, we follow the same procedure as Section \ref{sec: asymptotic results under the alternative} to compute the corresponding thresholds $c(\gamma, 0.05)$. Then, from the 101th observation, we continue to update the test statistic $w^2(m,k)S_{m,k}(\hat{\bm{\eta}})^T\bm{A}S_{m,k}(\hat{\bm{\eta}})$ until the test statistic exceeds $c(\gamma, 0.05)$. The time point is considered as the change point in this case.

\begin{table}[!h]
\centering
\begin{tabular}{l|lll}
\hline
$\gamma$     & 0          & 0.25       & 0.4        \\ \hline
Change point & 07/03/2021 & 07/02/2021 & 07/02/2021 \\ \hline
\end{tabular}
\caption{Detected change points based on the weight functions with different $\gamma$'s.}
\label{table: detected change points}
\end{table}

According to the detected change points in Table \ref{table: detected change points}, we can see that even with different $\gamma$'s, the detected change points are close and they are all at the beginning of July 2021. However, according to Arizona local news, the majority of the reports about Delta variant spreading in Arizona were after the detected change points, e.g. the report on 07/09/2021 \footnote{\url{https://www.azmirror.com/2021/07/09/deadly-delta-variant-of-covid-19-is-spreading-rapidly-in-arizona/}} and the report on 07/20/2021 \footnote{\url{https://www.azcentral.com/story/news/local/arizona-health/2021/07/21/delta-variant-covid-19-identified-dominant-az/8021859002/}}. Therefore, we can see that the change-point detection method is able to provide alerts for the spread of Delta variant prior to the news, and can be used in future Covid-19 postive rate monitoring and other related applications. 
\newpage
\section{Appendix}

\subsection{Proof of Lemma \ref{lemma:weak Feller}}
\label{app: weak feller chain}

For any $y = (x,\bm{W})$, the transition probability $P(y,O)$ is
\begin{equation}
\label{eq: decompose transition probability}
    \begin{aligned}
    P(y,O) &= P((X_{t+1},\bm{W}_{t+1}) \in O|(X_t,\bm{W}_t) = (x,w)) \\
    & = P(\bm{W}_{t+1} \in O_W|\bm{W}_t = w)P(X_{t+1} \in O_X|X_t = x, \bm{W}_{t+1} \in O_W) \\
    \end{aligned}
\end{equation}

As a product of two continuous functions with respect to $x$ and $w$ separately, we can see immediately that for any open set $O \in \mathcal{B}(\Omega)$, $P(y,O)$ is continuous on $y = (x,w)$, which implies lower continuity. Therefore, the process $\{\bm{Y}_t\}$ is a weak Feller chain. 

From Theorem 6.2.8 of \cite{meyn1993markov}, we can conclude that the compact state space, $\Omega$, is a petite set. The definition of petite sets is based on multiple other concepts, so we don't introduce it here. The details can be found in \cite{meyn1993markov}. What's noticeable is that for $\psi$-irreducible and aperiodic Markov chains, every petite set is a small set, see Theorem 5.5.7 of \cite{meyn1993markov}. Therefore, combining Lemma \ref{lemma:irreducibility&aperiodicity} and Lemma \ref{lemma:weak Feller}, we can conclude that the state space $\Omega$ is a small set.

\subsection{Proof of Theorem \ref{theorem:stationary&ergodic}}
\label{app: stationary ergodicity}
First, we prove the geometric ergodicity of the model. It is analogous to the proof of geometric ergodicity of Binomial AR(1) model in \cite{Liu_2025}. According to Theorem 15.0.1 of \cite{meyn1993markov}, for a $\psi$-irreducible and aperiodic Markov chain, if there exists a petite set $C$, constant $b <\infty$, $\beta>0$ and a function $V \geq 1$ finite at some one $y_0 \in \Omega$ such that
\begin{equation*}
    \Delta V(x) \leq -\beta V(x) +b \mathbbm{1}_C(x), x \in \Omega,
\end{equation*}
then the process is geometrically ergodic, where $\Delta V(x)$, as defined on page 558 of \cite{meyn1993markov}, is 
\begin{equation*}
    \Delta V(x) = \int P(x,\text{d}y)V(y)-V(x),
\end{equation*}
and $P(x, \text{d}y)$ is the one-step transition pdf from $x$ to $y$.

From Lemma \ref{lemma:weak Feller} we have proven that the state space $\Omega$ is a small (also petite) set. Define $V(\cdot)$ as the constant function $\Omega \rightarrow b \geq 1$, and let $\beta = 1$, then the drift condition can be easily satisfied.
\begin{equation*}
    \begin{aligned}
    \Delta V(x) = & \int P(x,\text{d}y)V(y)-V(x) \\
             \leq & (\underset{y \in \Omega} {\max}V(y))-V(x) \\
              = & b-V(x) \\
              = & -V(x)+b\mathbbm{1}_{\Omega}(x), \text{ for any }x \in \Omega.
    \end{aligned}
\end{equation*}
So $\{\bm{Y}_t\}$ is geometrically ergodic. Also, according to Theorem 15.0.1 in \cite{meyn1993markov},  $\psi$-irreducibility, aperiodicity and the drift condition together imply the existence of the unconditional probability of the process, say $\pi$, such that
\begin{equation*}
    \pi(A) = P(\bm{Y}_t \in A), \text{ for any } t\in \mathbbm{Z} \text{ and }A \in \mathcal{B}(\Omega).
\end{equation*}
From the modeling specification (\ref{equation: model definition}) and (\ref{eq: link function of GBETAAR}), we can conclude that the one-step transition probability is free from $t$. Then for any $A_0, A \in \mathcal{B}(\Omega)$ and $n\geq 1$, 
\begin{equation*}
    P^n(A_0,A) = P(\bm{Y}_{t+n} \in A|\bm{Y}_t \in A_0) = \int_A \int_{\Omega}\cdots \int_{\Omega}P(A_0,\text{d}y_{t+1})P(\text{d}y_{t+1},\text{d}y_{t+2})\cdots P(\text{d}y_{t+n-1},\text{d}y_{t+n})
\end{equation*}
is free from $t$. So for any random vectors $(\bm{Y}_{t_1},...,\bm{Y}_{t_n})$ and $(\bm{Y}_{t_1+\tau},...,\bm{Y}_{t_n+\tau})$, the joint pdf of the random vectors
\begin{equation*}
\begin{aligned}
    & P(\bm{Y}_{t_n} \in A_{n}, \bm{Y}_{t_{n-1}} \in A_{n-1},\cdots,\bm{Y}_{t_1}=A_1) \\
    = &P^{t_{n}-t_{n-1}}(A_{n-1},A_n)\times \cdots \times P^{t_{2}-t_{1}}(A_1,A_2)\pi(A_1)\\
    = &P(\bm{Y}_{t_n+\tau} \in A_{n}, \bm{Y}_{t_{n-1}+\tau} \in A_{n-1},\cdots,\bm{Y}_{t_1+\tau}=A_1).
\end{aligned}
\end{equation*}
Therefore, the process is strictly stationary.

\subsection{Proof of Lemma \ref{lemma: unique of eta_0}}
\label{app: unique of eta_0}
The proof is similar with the proof of identifiability of the true parameter (part (b) in the proof of Theorem 4.1) in \cite{gorgi2019bnb}. 

For any $x>0$, the following inequality holds
\begin{equation*}
    \log x \leq x-1.
\end{equation*}
$"="$ holds iff $x=1.$
Therefore, for any $\bm{\eta} \in \Xi$ and $\bm{\eta} \neq \bm{\eta}_0$, 
\begin{equation}
    \log \frac{f_{\bm{\eta}}(X_t|\mathcal{C}_t)}{f_{\bm{\eta}_0}(X_t|\mathcal{C}_t)} \leq \frac{f_{\bm{\eta}}(X_t|\mathcal{C}_t)}{f_{\bm{\eta}_0}(X_t|\mathcal{C}_t)}-1.
\end{equation}
In the meantime, with $\bm{\eta} \neq \bm{\eta}_0$, $\frac{f_{\bm{\eta}}(X_t|\mathcal{C}_t)}{f_{\bm{\eta}_0}(X_t|\mathcal{C}_t)} \neq 1$ with positive probability. Therefore, 
\begin{equation*}
    \begin{aligned}
    &pl(\bm{\eta})-pl(\bm{\eta}_0)  \\
    =&\mathbbm{E}\Big(\log f_{\bm{\eta}}(X_t|\mathcal{C}_t)\Big)-\mathbbm{E}\Big(\log f_{\bm{\eta}_0}(X_t|\mathcal{C}_t)\Big) \\
    =& \mathbbm{E}\Big(\log \frac{f_{\bm{\eta}}(X_t|\mathcal{C}_t)}{f_{\bm{\eta}_0}(X_t|\mathcal{C}_t)}\Big) \\
    <& \mathbbm{E}\Big(\frac{f_{\bm{\eta}}(X_t|\mathcal{C}_t)}{f_{\bm{\eta}_0}(X_t|\mathcal{C}_t)}\Big)-1 \\
    =& \mathbbm{E}\Big(\mathbbm{E}(\frac{f_{\bm{\eta}}(X_t|\mathcal{C}_t)}{f_{\bm{\eta}_0}(X_t|\mathcal{C}_t)}|\mathcal{C}_t)\Big)-1 \\
    =& \mathbbm{E}\Big(\int_0^1 \frac{f_{\bm{\eta}}(x|\mathcal{C}_t)}{f_{\bm{\eta}_0}(x|\mathcal{C}_t)}f_{\bm{\eta}_0}(x|\mathcal{C}_t) \text{d}x\Big)-1 \\
    =& \mathbbm{E}\Big(\int_0^1 f_{\bm{\eta}}(x|\mathcal{C}_t)\text{d}x \Big)-1\\
    =&0.
    \end{aligned}
\end{equation*}
\subsection{Proof of the bounded fourth moment of $G$}
Here we prove the second element of $G(X_t, \bm{\eta})$ has bounded fourth moment for any fixed $\bm{\eta} \in \Xi$. That is, $\mathbbm{E}[|\tau(X_t^*-\mu_t^*)\mu_t(1-\mu_t)|^4]<\infty$, where $X^*_t = \log(\frac{X_t}{1-X_t})$ and $\mu_t^* = \psi(\tau \mu_t)-\psi(\tau(1-\mu_t))$. As shown in Lemma \ref{lemma: G is MDS}, $\mathbbm{E}(X^*_t) = \mu^*_t.$ By using the same idea, we can also prove other elements of $G$ have bounded fourth moments. 

For a random variable $X\sim $Beta$(\alpha,\beta)$, in (\ref{eq: cgf for logX/(1-X)}), we have calculated the cumulant-generating function of $\log(\frac{X}{1-X})$, i.e. $K_{\log(\frac{X}{1-X})}(t)$. Denote the $j^{\text{th}}$ cumulant as $\kappa_j$, denote $\log(\frac{X}{1-X})$ as $X^*$ and denote $\mathbbm{E}(X^*) = \psi(\alpha)-\psi(\beta)$ as $\mu^*$ (the notations are similar with the ones for $X_t$). According to \cite{dodge1999complications} and the expression of $K_{\log(\frac{X}{1-X})}(t)$, 
\begin{equation*}
\begin{aligned}
    \kappa_2 &= \mathbbm{E}((X^*- \mu^*)^2) = \frac{\text{d}^2}{\text{d}t^2}K_{\log(\frac{X}{1-X})}(t)|_{t=0} = \psi'(\alpha)+ \psi'(\beta). \\
    \kappa_4 &= \mathbbm{E}((X^*- \mu^*)^4)-3\mathbbm{E}((X^*- \mu^*)^2)^2 = \frac{\text{d}^4}{\text{d}t^4}K_{\log(\frac{X}{1-X})}(t)|_{t=0} = \psi^{(3)}(\alpha)+ \psi^{(3)}(\beta).
\end{aligned}
\end{equation*}
So the fourth central moment of $X^*$ is 
\begin{equation*}
    \mathbbm{E}((X^*- \mu^*)^4) = \kappa_4+3 \kappa_2^2 = \psi^{(3)}(\alpha)+ \psi^{(3)}(\beta)+3 (\psi'(\alpha)+ \psi'(\beta))^2.
\end{equation*}
That is, given the information set $\mathcal{C}_t$ (so $\mu_t$ has been determined by the past information), the Beta random variable $X_t\sim \text{Beta}(\tau \mu_t, \tau (1-\mu_t))$ has 
\begin{equation*}
   \mathbbm{E}((X_t^*-\mu_t^*)^4|\mathcal{C}_t) = \psi^{(3)}(\tau \mu_t)+ \psi^{(3)}(\tau (1-\mu_t))+3 (\psi'(\tau \mu_t)+ \psi'(\tau (1-\mu_t)))^2.
\end{equation*}
According to the series representation of $\psi'$ and $\psi^{(3)}$ shown in \cite{batir2007some}, $\psi'(x)>0$ and $\psi^{(3)}(x)>0$ for $x \in (0,\infty)$, and they are both continuous and decreasing functions. Since the state space of the Markov chain, $\Omega$, is compact, for any fixed parameter vector $\bm{\eta} \in \Xi$, we can find the lower bound and the upper bound of $\bm{\beta}^T \mathbbm{Z}_{t-1}$ ($\bm{\beta}$ is the parameter vector excluding the first element ($\tau$) in $\bm{\eta}$), which are free from $t$. Determined by $\bm{\beta}^T \mathbbm{Z}_{t-1}$ through (\ref{eq: link function of GBETAAR}), $\mu_t$, therefore, has the lower bound $l_{\bm{\eta}}>0$, and the upper bound, $u_{\bm{\eta}}<1$. $l_{\bm{\eta}}$ and $u_{\bm{\eta}}$ are free from $t$. Then, considering $\psi'$ and $\psi^{(3)}$ are both decreasing,
\begin{equation*}
    \begin{aligned}
    &\psi'(\tau\mu_t) \leq \psi'(\tau l_{\bm{\eta}}), \quad \psi^{(3)}(\tau\mu_t) \leq \psi^{(3)}(\tau l_{\bm{\eta}}), \\
        &\psi'(\tau(1-\mu_t)) \leq \psi'(\tau (1-u_{\bm{\eta}})), \quad \psi^{(3)}(\tau(1-\mu_t)) \leq \psi^{(3)}(\tau (1-u_{\bm{\eta}})).
    \end{aligned}
\end{equation*}
Therefore, 
\begin{equation*}
\begin{aligned}
   \mathbbm{E}((X_t^*-\mu_t^*)^4|\mathcal{C}_t)& \leq \psi^{(3)}(\tau l_{\bm{\eta}})+\psi^{(3)}(\tau (1-u_{\bm{\eta}}))+3(\psi'(\tau l_{\bm{\eta}})+\psi^{(3)}(\tau (1-u_{\bm{\eta}})))^2 \\
   &\overset{\Delta}{=} b(\bm{\eta}) \\
   &<\infty.
\end{aligned}
\end{equation*}
Given the information, 
\begin{equation*}
    \begin{aligned}
    \mathbbm{E}[|\tau(X_t^*-\mu_t^*)\mu_t(1-\mu_t)|^4]&=\tau^4\mathbbm{E}[\mathbbm{E}((X_t^*-\mu_t^*)^4\mu^4_t(1-\mu_t)^4|\mathcal{C}_t)]\\
    &=\tau^4 \mathbbm{E}[\mu^4_t(1-\mu_t)^4\mathbbm{E}((X_t^*-\mu_t^*)^4|\mathcal{C}_t)] \\
    &<\tau^4 \times b(\bm{\eta})\\
    &< \infty.
    \end{aligned}
\end{equation*}
\label{proof: bounded fourth moment}

\subsection{Proof of Theorem \ref{theorem: CLT}}
\label{proof: CLT}
According to (\ref{eq: Sm=0}), $S_m(\hat{\bm{\eta}}) = 0.$
By Taylor expansion, we can find $\Tilde{\bm{\eta}} = (\Tilde{\tau},\Tilde{\bm{\beta}})$ between $\hat{\bm{\eta}}$ and $\bm{\eta}_0$ such that 
\begin{equation}
    \label{eq: taylor expansion}
    S_m(\hat{\bm{\eta}}) = S_m(\bm{\eta}_0)+\triangledown S_m(\Tilde{\bm{\eta}})(\hat{\bm{\eta}}-\bm{\eta}_0) = 0.
\end{equation}

In order to prove the asymptotic normality of $\hat{\bm{\eta}}$, we first show the invertibility of $\triangledown S_m(\Tilde{\bm{\eta}})$ as $m \rightarrow \infty$. The strong consistency of $\hat{\bm{\eta}}$ proven in Theorem \ref{theorem: consistency} implies 
\begin{equation*}
    \Tilde{\bm{\eta}} \overset{a.s.}{\rightarrow} \bm{\eta}_0.
\end{equation*}
We have provided the closed-form expression of $\triangledown S_m(\bm{\eta})$ in Section \ref{subsec: parameter estimation}, from which we can conclude that it is a continuous function with respect to $\bm{\eta}$. Based on Continuous Mapping Theorem and Birkhoff's Ergodic Theorem,
\begin{equation}
\label{eq: ergodic of S_m(tilde(eta))}
    \frac{1}{m}\triangledown S_m(\Tilde{\bm{\eta}})\overset{a.s.}{\rightarrow}\mathbbm{E}\Big(\triangledown G(X_1,\bm{\eta}_0)\Big).
\end{equation}
Given that $G(X_1, \bm{\eta}_0)$ is the first derivative of the log-likelihood function and the continuity of $G$ with respect to $\bm{\eta}$, the equation 
\begin{equation*}
    \mathbbm{E}\Big(\triangledown G(X_1,\bm{\eta}_0)\Big) = -\mathbbm{E}\Big(G(X_1,\bm{\eta}_0)G(X_1,\bm{\eta}_0)^T\Big)
\end{equation*}
can be shown by following the standard procedure. Also, following the same method as Lemma A.5. of \cite{gorgi2019bnb}, $-\mathbbm{E}\Big(\triangledown G(X_1,\bm{\eta}_0)\Big)$ is a positive definite matrix. Therefore, with probability one, $-\frac{1}{m}\triangledown S_m(\Tilde{\bm{\eta}})$ is positive definite and therefore invertible for all $m$ sufficiently large.

As shown in Lemma \ref{lemma: G is MDS}, $G(X_t,\bm{\eta}_0)$ is a MDS and therefore $S_m(\bm{\eta}_0) = \sum_{t=1}^m G(X_t, \bm{\eta}_0)$ is a martingale. We first prove the martingale central limit theorem (MCLT) holds on $S_m(\bm{\eta}_0).$ 

We denote $\mathbbm{E}(S_m(\bm{\eta}_0)S_m(\bm{\eta}_0)^T)$ as $I_m^2(\bm{\eta}_0)$. According to \cite{brown1971martingale}, $$I_m^2(\bm{\eta}_0) = m\mathbbm{E}(G(X_1,\bm{\eta}_0)G(X_1,\bm{\eta}_0)^T) = -m \mathbbm{E}\Big(\triangledown G(X_1,\bm{\eta}_0)\Big).$$
We can easily get the same conclusion considering the $S_m(\bm{\eta}_0)$ is a martingale.
In \cite{brown1971martingale}, the author proposes the conditions under which the univariate MCLT holds. For ergodic univariate martingale processes, the MCLT holds when the Lindeberg condition holds. In order to apply the theorem directly, we create the univariate martingale $\bm{\lambda}^T S_m(\bm{\eta}_0)$ in which $\bm{\lambda}$ can be any nonzero $(l+3)$-dimensional vector and first prove the MCLT holds for any $\bm{\lambda}^T S_m(\bm{\eta}_0)$. It is easy to see that $\mathbbm{E}((\bm{\lambda}^T S_m(\bm{\eta}_0))(\bm{\lambda}^T S_m(\bm{\eta}_0))^T)$ is $\bm{\lambda}^TI_m^2(\bm{\eta}_0)\bm{\lambda}.$ We denote it as $I_{m,\lambda}^2(\bm{\eta}_0).$ Then we first show that the Lindeberg condition holds for $\bm{\lambda}^TS_m(\bm{\eta}_0)$.
That is, for any $\epsilon>0$, as $m \rightarrow \infty$,
\begin{equation*}
\begin{aligned}
&\frac{1}{I_{m,\lambda}^2(\bm{\eta}_0)}\sum_{t=1}^m \mathbbm{E}((\bm{\lambda}^TG(X_t,\bm{\eta}_0))^2\mathbbm{1}_{(|\bm{\lambda}^T G|\geq \epsilon I_{m,\lambda}(\bm{\eta}_0))}) \\
\leq &\frac{1}{I_{m,\lambda}^2(\bm{\eta}_0)}\sum_{t=1}^m \mathbbm{E}((\bm{\lambda}^T G)^4)^{\frac{1}{2}}P(|\bm{\lambda}^T G|\geq \epsilon I_{m,\lambda}(\bm{\eta}_0))^{\frac{1}{2}}, \text{by Cauchy-Schwarz inequality} \\
\leq & \frac{1}{I_{m,\lambda}^2(\bm{\eta}_0)}\sum_{t=1}^m \mathbbm{E}((\bm{\lambda}^T G)^4)^{\frac{1}{2}}\Big [\frac{\mathbbm{E}(|\bm{\lambda}^TG|^4)}{(\epsilon^4 I_{m,\lambda}^4(\bm{\eta}_0))}\Big]^{\frac{1}{2}},\text{ by Markov's inequality} \\
= & \frac{1}{(I_{m,\lambda}^2(\bm{\eta}_0))^2\epsilon^2}\sum_{t=1}^m\mathbbm{E}((\bm{\lambda}^T G)^4),\\
= & \frac{m\mathbbm{E}((\bm{\lambda}^TG)^4)}{m^2 (-\bm{\lambda}^T \mathbbm{E}\triangledown G\bm{\lambda})^2 \epsilon^2} \\
{\rightarrow}& 0.
\end{aligned}    
\end{equation*}
Therefore, based on Theorem 2 of \cite{brown1971martingale}, 
\begin{equation}
    \label{eq: MCLT of the univariate martingale}
    \begin{aligned}
    &\bm{\lambda}^T S_m(\bm{\eta}_0)/\sqrt{(I^2_{m,\lambda}(\bm{\eta}_0))} \\
    =&\frac{1}{\sqrt{m}}\bm{\lambda}^T S_m(\bm{\eta}_0)(-\bm{\lambda}^T \mathbbm{E}\triangledown G(X_1,\bm{\eta}_0) \bm{\lambda})^{-1/2} \\
    \overset{\mathcal{D}}{\rightarrow}&N(0,1).    
    \end{aligned}
\end{equation}
With Cramer-Wold device,
\begin{equation}
\label{eq: normality of S_m(beta)}
    \frac{1}{\sqrt{m}}S_m(\bm{\eta}_0) \overset{\mathcal{D}}{\rightarrow} N(0,-\mathbbm{E}\triangledown G(X_1,\bm{\eta}_0))
\end{equation}

Combining (\ref{eq: taylor expansion}), (\ref{eq: ergodic of S_m(tilde(eta))}) and (\ref{eq: normality of S_m(beta)}), 
\begin{equation*}
    \begin{aligned}   
    \sqrt{m}(\hat{\bm{\eta}}-\bm{\eta}_0) =&-\sqrt{m}\triangledown S_m(\Tilde{\bm{\eta}})^{-1}S_m(\bm{\eta}_0)\\
    =& -(\frac{1}{m}\triangledown S_m(\Tilde{\bm{\eta}}))^{-1}\frac{1}{\sqrt{m}}S_m(\bm{\eta}_0) \\
    \overset{\mathcal{D}}{\rightarrow}& N(0,(-\mathbbm{E}\triangledown G(X_1,\bm{\eta}_0))^{-1}).
    \end{aligned}
\end{equation*}

\subsection{Proof of Lemma \ref{proposition: CUSUM approximation}}
\label{app: proof of CUSUM approximation}
Here we first show that in a compact neighborhood of the true parameter $\bm{\eta}_0$, $\delta_{\bm{\eta}_0} \subset \Xi$,
\begin{equation*}
\underset{\bm{\eta} \in \delta_{\bm{\eta}_0}}{\sup} ||\triangledown^2 G_j(X_1,\bm{\eta})||_{\infty} =O_p(1), \text{ for } j \in \{1,2,...,l+3\} \text{ and } X_1 \text{ follows model (\ref{equation: model definition}) with parameter } \bm{\eta},
\end{equation*}
where $\triangledown^2 G_j(X_1,\bm{\eta})$ is the Hessian matrix of the $j^{\text{th}}$ element of $G$ and $||.||_{\infty}$ is the maximum of the absolute values of all the elements of a matrix.

From (\ref{eq: score vector}) we know that every element of $G$ is a continuous function with respect to $\bm{\eta} = (\tau,\bm{\beta}) \in (0,\infty)\times \mathbbm{R}^{l+2}$ and $X_1$. Meanwhile, the digamma function $\psi$ and the inverse of the logit link function are both twice continuously differentiable, so every element in $\triangledown^2 G_j(X_1,\bm{\eta})$ is continuous with respect to $\bm{\eta}$ and $X_1$. From the closed-form expression of $G$ shown in (\ref{eq: score vector}), we can see that every element in $\triangledown^2 G_j(X_1,\bm{\eta})$ is a continuous function of $X^*_1 = \log(X_1/(1-X_1))$. So we denote the $(i_1,i_2)^{\text{th}}$ element of $\triangledown^2 G_j(X_1,\bm{\eta})$ as $\triangledown^2 G_j^{i_1,i_2}(X^*_1,\bm{\eta}), i_1, i_2 \in \{1,2,...,l+3\}$.

First, given that $X_1$ follows the generalized Beta AR(1) model, for any $\epsilon>0$, we can always find a $M>0$ so that
\begin{equation*}
\label{eq: bounded transformed X1}
    P(\underset{\bm{\eta} \in \delta_{\bm{\eta}_0}}{\sup}|X^*_1|>M) <\epsilon,
\end{equation*}
where $\delta_{\bm{\eta}_0}$ is the fixed compact neighborhood of the true parameter $\bm{\eta}_0$. Equivalently,
\begin{equation}
\label{eq: bounded transformed X2}
        P(\underset{\bm{\eta} \in \delta_{\bm{\eta}_0}}{\sup}|X^*_1| \leq M) \geq 1-\epsilon.
\end{equation}
Then, taking advantage of the continuity of $\triangledown^2 G_j^{i_1,i_2}(X^*_1,\bm{\eta})$, for every $(i_1,i_2)$, we can find the corresponding $M_{i_1,i_2}(\bm{\eta})$ given the upper bound $M$ for $X^*_1$ so that 
\begin{equation}
\label{eq: find M_i}
    |\triangledown^2 G_j^{i_1,i_2}(X^*_1,\bm{\eta})| \leq M_{i_1,i_2}(\bm{\eta}), \text{ for }i_1,i_2 \in \{1,2,...,l+3\}.
\end{equation}
By extreme value theorem, every $M_{i_1,i_2}(\bm{\eta})$ is bounded in $\delta_{\bm{\eta}_0}$. Therefore, we are able to find $M^* = \underset{\bm{\eta} \in \delta_{\bm{\eta}_0}}{\sup}(M_{1,1}(\bm{\eta}),M_{1,2}(\bm{\eta}),...,M_{l+3,l+3}(\bm{\eta}))$ so that $\underset{\bm{\eta} \in \delta_{\bm{\eta}_0}}{\sup} ||\triangledown^2 G_j(X_1,\bm{\eta})||_{\infty} \leq M^*.$ That is, given the $M$ enabling (\ref{eq: bounded transformed X2}), we can always find the corresponding $M^*$ as the upper bound for $\underset{\bm{\eta} \in \delta_{\bm{\eta}_0}}{\sup} ||\triangledown^2 G_j(X_1,\bm{\eta})||_{\infty}$. Combining (\ref{eq: bounded transformed X2}) and (\ref{eq: find M_i}), it holds that 
\begin{equation}
   P(\underset{\bm{\eta} \in \delta_{\bm{\eta}_0}}{\sup} ||\triangledown^2 G_j(X_1,\bm{\eta})||_{\infty} \leq M^*) \geq 1-\epsilon.
\end{equation}
Equivalently, 
\begin{equation}
    P(\underset{\bm{\eta} \in \delta_{\bm{\eta}_0}}{\sup} ||\triangledown^2 G_j(X_1,\bm{\eta})||_{\infty} >M^*) <\epsilon.
\end{equation}
So 
\begin{equation*}
    \underset{\bm{\eta} \in \delta_{\bm{\eta}_0}}{\sup} ||\triangledown^2 G_j(X_1,\bm{\eta})||_{\infty} =O_p(1)
\end{equation*}
holds for $j \in \{1,2,...,l+3\}$. 

According to Theorem \ref{theorem:stationary&ergodic}, we have derived the properties that the Markovian process $\{\bm{Y}_t\}$ is strictly stationary and geometrically ergodic. Also, Theorem \ref{theorem: consistency} and Theorem \ref{theorem: CLT} show the consistency and asymptotic normality of the PMLE $\hat{\bm{\eta}}$. Then, the proof of the lemma follows based on the properties.
\subsection{Proof of Proposition \ref{proposition: wiener processes}}
\label{app: proof of wiener process}

In order to prove that the general state space Markov chain $\{\bm{Y}_t\}$ is strong mixing, we first need to introduce the relationship between the absolute regularity ($\beta$-mixing) coefficient $\beta(n)$ and the strong mixing coefficient $\alpha(n)$, where $n$ is the distance between two disjoint information sets. For example, the distance between the information set recording outputs before $\{\bm{Y}_t\}$ and the information set recording outputs after $\{\bm{Y}_{t+n}\}$ is $n$.
The definition of absolute regularity condition can be found from (1.5) of \cite{bradley2005basic} (on page 108). In a nutshell, absolute regularity implies strong mixing. Mathematically, from (1.11) of \cite{bradley2005basic}, we have 
\begin{equation}
\label{eq: mixing inequality}
    2\alpha(n) \leq \beta(n).
\end{equation}
Then by Theorem 3.7 of \cite{bradley2005basic} (on page 121), if a strictly stationary Markov chain defined on a general state space is geometrically ergodic, the process has $\beta(n)\rightarrow 0$ at least exponentially fast as $n \rightarrow \infty$. That is, there exists a $\theta$ such that 
\begin{equation*}
    \beta(n) =O_p(e^{-\theta n}).
\end{equation*}
Due to the inequality (\ref{eq: mixing inequality}), $\alpha(n) = O_p(e^{-\theta n}).$ Since $G(X_{t+n},\bm{\eta}_0)$ is in the $\sigma$-field $\mathcal{C}_{t+n-1},$ the strong mixing coefficient for $G(X_t,\bm{\eta}_0)$, $\alpha_G(n) = \alpha(n-1) = O_p(e^{-\theta n})$. 

From Lemma \ref{lemma: G is MDS} we have shown that $G(X_t,\bm{\eta}_0)$ is a MDS. Therefore, for any $t>1$, $$\mathbbm{E} \big (G(X_1,\bm{\eta}_0)G^T(X_t,\bm{\eta}_0)\Big) = \mathbbm{E} \big(\mathbbm{E}(G(X_1,\bm{\eta}_0)G^T(X_t,\bm{\eta}_0)|\mathcal{C}_{t-1})\Big) = 0.$$ 
From the proof of Theorem \ref{theorem: CLT}, 
\begin{equation*}
    \mathbbm{E}\big(G(X_1,\bm{\eta}_0)G^T(X_1,\bm{\eta}_0)\big) = -\mathbbm{E}\Big(\triangledown G(X_1,\bm{\eta}_0)\Big).
\end{equation*}
Then, the covariance matrix
\begin{equation*}
\begin{aligned}
 \bm{\Sigma} &= \mathbbm{E}\big(G(X_1,\bm{\eta}_0)G^T(X_1,\bm{\eta}_0)\big)+2\sum_{t=2}^{\infty}\mathbbm{E} \big(G(X_1,\bm{\eta}_0)G^T(X_t,\bm{\eta}_0) \big)\\
 &=-\mathbbm{E}\Big(\triangledown G(X_1,\bm{\eta}_0)\Big). 
\end{aligned}
\end{equation*}
The closed-form of $-\mathbbm{E}\Big(\triangledown G(X_1,\bm{\eta}_0)\Big)$ has been shown in Section \ref{subsec: Closed-form expressions of the partial log-likelihood function, the score vector and the Hessian matrix}. Relying on the property of $\log(\frac{X_t}{1-X_t})$ we have discussed in the proof of Lemma \ref{lemma: bounded fourth moment of G}, it is easy to see that $\bm{\Sigma} = -\mathbbm{E}\Big(\triangledown G(X_1,\bm{\eta}_0)\Big) <\infty.$
As a strong mixing process with bounded covariance matrix, following Theorem 1 of \cite{oodaira1972functional}, 
\begin{equation*}
    \left\{\frac{1}{\sqrt{m}}\sum_{t=1}^{\floor{ms}}G(X_t,\bm{\eta}_0)):0< s\leq N \right\} \overset{\mathcal{D}}{\rightarrow}W(s) \text{ as }m \rightarrow \infty,
\end{equation*}
where $W(s)$ is a $(l+3)$-dimensional Wiener process with the covariance matrix $\bm{\Sigma}$. Analogous to the definition of a Wiener process in \cite{oodaira1972functional}, $W(s)$ is defined in $D = D[0,N]$, the space of $(l+3)$-dimensional continuous function with $s \in [0,N].$

\subsection{Proof of Theorem \ref{theorem: null}}
\label{app: proof of the null}
We continue to use the notation introduced in \cite{kirch2015use} that $Z^TAZ:=||Z||^2_A$. \\
By using Lemma \ref{proposition: CUSUM approximation}, the test statistic
\begin{equation}
\begin{aligned}
\label{equation: close-end test statistics 1}
        &\underset{1\leq k \leq Nm}{\sup} \omega(m,k)^2 ||S_{m,k}(\hat{\bm{\eta}})||^2_A \\
        =&\underset{1\leq k \leq Nm}{\sup} \omega(m,k)^2
        ||\sum_{t=m+1}^{m+k}G(X_t,\bm{\eta}_0)-\frac{k}{m}\sum_{t=1}^m G(X_t,\bm{\eta}_0)||^2_A+o_p(1).\\
\end{aligned}
\end{equation}

From Proposition \ref{proposition: wiener processes} we know that with $k = \lfloor ms$,
\begin{equation}
\label{equation: close-end G1}
    \frac{1}{\sqrt{m}}\sum_{t=m+1}^{m+k}G(X_t,\bm{\eta}_0) = \frac{1}{\sqrt{m}}(\sum_{t=1}^{m+k}G(X_t,\bm{\eta}_0)-\sum_{t=1}^m G(X_t,\bm{\eta}_0)) \overset{\mathcal{D}}{\rightarrow} W_1(k/m):=W_1(s),
\end{equation}
and $W_1(s)$ is independent with the $\sigma$-field $\mathcal{F}_0^m$ considering that as a Wiener process, $W_1(.)$ has independent increments. \\
Also, 
\begin{equation}
\label{equation: close-end G2} 
    \frac{1}{\sqrt{m}}\frac{k}{m} \sum_{t=1}^m G(X_t,\bm{\eta}_0) \overset{\mathcal{D}}{\rightarrow} sW_2(1)
\end{equation}
can be drawn from Proposition \ref{proposition: wiener processes}. Because $W_2(1) \in \mathcal{F}_0^m$, $W_1(.) \indep W_2(.)$ holds. \\
The continuity of $\rho^2(t)||W_1(t)-tW_2(1)||^2_A$ on $s \in [0,N]$, together with (\ref{equation: close-end test statistics 1}), (\ref{equation: close-end G1}) and (\ref{equation: close-end G2}) yield that
\begin{equation*}
\begin{aligned}
        \label{equation: close-end test statistics 2}
    &\quad\underset{1\leq k \leq Nm}{\sup}\omega(m,k)^2 ||S_{m,k}(\hat{\bm{\eta}})||^2_A \\
    &\overset{\mathcal{D}}{\rightarrow} \underset{0<s\leq N}{\sup} \rho^2(s)||W_1(s)-sW_2(1)||^2_A.
\end{aligned}
\end{equation*}
So for a given threshold $c$,
\begin{equation*}
\begin{aligned}
    \underset{m \rightarrow \infty}{\lim}& P\left(\underset{1\leq k \leq Nm}{\sup} w^2(m,k)S_{m,k}(\hat{\bm{\eta}})^T\bm{A}S_{m,k}(\hat{\bm{\eta}}) \leq c    \right) \\
     = &P\left( \underset{0<s\leq N}{\sup} \rho^2(s)(W_1(s)-sW_2(1))^T\bm{A}(W_1(s)-sW_2(1))\leq c \right).
\end{aligned}
\end{equation*}

\subsection{Proof of Theorem \ref{theorem: power analysis}}
\label{app: proof of power analysis}
We denote the change point after the first $m$ non-contaminated observations as $k^*$, and use $\Tilde{k} = \lfloor m x_0 \rfloor$ to denote a data point after $k^*$. Then 
\begin{equation*}
\begin{aligned}
        S_{m,\Tilde{k}}(\hat{\bm{\eta}}) &= \sum_{t=m+1}^{m+k^*} G(X_t,\hat{\tau}, \hat{\bm{\beta}})+ \sum_{t=m+k^*+1}^{m+\Tilde{k}} G(X_t^*,\hat{\bm{\eta}})\\
                            &=: \bm{S}_{H_0}(m,k^*)+\bm{S}_{H_a}(m+k^*,\Tilde{k}-k^*)
\end{aligned}
\end{equation*}
As $m \rightarrow \infty$, $k^* \rightarrow \infty$, due to ergodicity of $\{X_t\}$ and consistency of $\hat{\bm{\beta}}$,
\begin{equation*}
    \frac{1}{k^*} \bm{S}_{H_0}(m,k^*) \overset{P}{\rightarrow}0.
\end{equation*}
Based on Assumption \ref{assumption: power analysis}(a),
\begin{equation}
\label{eq: power proof 1}
    \frac{1}{m} \bm{S}_{H_0}(m,k^*) \overset{P}{\rightarrow}0.
\end{equation}
After the change point $k^*$, under Assumption \ref{assumption: power analysis}(c), 
\begin{equation}
\label{eq: power proof 2}
    \frac{1}{m}\bm{S}_{H_a}(m+k^*,\Tilde{k}-k^*) = \frac{x_0-\nu}{(x_0-\nu)m} \sum_{t = m+k^*+1}^{m+\Tilde{k}}G(X_t^*,\hat{\tau}, \hat{\bm{\beta}}) \overset{P}{\rightarrow} (x_0-\nu) \bm{E}_{G,H_a}.
\end{equation}
Combining (\ref{eq: power proof 1}) and (\ref{eq: power proof 2}), as $m \rightarrow \infty$,
\begin{equation*}
    \begin{aligned}
    & \underset{1\leq k \leq Nm}{\sup}w^2(m,k)S_{m,k}(\hat{\bm{\eta}})^T\bm{A}S_{m,k}(\hat{\bm{\eta}}) \\
    \geq & \frac{1}{m}\rho^2(x_0) (\bm{S}_{H_0}+\bm{S}_{H_a})^T \bm{A} (\bm{S}_{H_0}+\bm{S}_{H_a})\\
    =&m \rho^2(x_0) (\frac{1}{m}\bm{S}_{H_0}+\frac{1}{m}\bm{S}_{H_a})^T \bm{A} (\frac{1}{m}\bm{S}_{H_0}+\frac{1}{m}\bm{S}_{H_a}) \\
     = & m\rho^2(x_0)(x_0-\nu)^2 (\bm{E}_{G,H_a}+o_p(1))^T \bm{A} (\bm{E}_{G,H_a}+o_p(1)) \\
     \rightarrow & \infty.
    \end{aligned}
\end{equation*}
Therefore, with any fixed threshold $c$,
\begin{equation*}
    \underset{m \rightarrow \infty}{\lim} P\left(\underset{1\leq k \leq Nm}{\sup}w^2(m,k)S_{m,k}(\hat{\bm{\eta}})^T\bm{A}S_{m,k}(\hat{\bm{\eta}}) \geq c  |H_a  \right) = 1.
\end{equation*}

\bibliography{arxiv} 
\bibliographystyle{ieeetr}
\end{document}